\documentclass[journal]{IEEEtran}

\usepackage[colorlinks,
linkcolor=red,
anchorcolor=blue,
citecolor=red
]{hyperref}
\usepackage{amsmath}
\usepackage{amssymb}
\usepackage{bm} 
\usepackage{graphicx} 
\usepackage{epstopdf}
\usepackage{enumerate}
\usepackage{booktabs} 
\usepackage{multirow}
\usepackage{threeparttable} 
 
\usepackage{color} 
\usepackage{float} 
\usepackage{cite}
\usepackage{subfigure}
\usepackage{hhline}

\usepackage{graphicx}
\DeclareGraphicsExtensions{.pdf,.jpeg,.png}
\usepackage{epstopdf}

\newtheorem{theorem}{Theorem}
\newtheorem{proof}{Proof}

\newtheorem{remark}{Remark}

\newtheorem{assumption}{Assumption}

\newtheorem{proposition}{Proposition}

\begin{document}

\title{Balanced control between performance and saturation for constrained nonlinear systems}

\author{Peng Wang,~ 
        Haibin Wang,~
        Shuzhi~Sam~Ge,~
        \IEEEmembership{Fellow,~IEEE}~
        and Xiaobing~Zhang 
 
\thanks{Peng Wang and Xiaobing Zhang are with the School of Energy and Power Engineering, Nanjing
University of Science and Technology, Nanjing 210094, China (e-mail:
benjywp711@gmail.com, zhangxb680504@163.com).}
\thanks{Haibin Wang is with 
College of Automation, Harbin University of Science and Technology, Harbin, 150080, China (e-mail: wanghaibinheu@hrbeu.edu.cn)}
\thanks{Shuzhi Sam Ge is with the Department of Electrical and Computer Engineering, National University of Singapore, Singapore 117576 (e-mail: samge@nus.edu.sg).}
}


\maketitle

\begin{abstract}
This paper addresses the balanced control between performance and saturation for a class of constrained nonlinear systems, including the branches: balanced command filtered backstepping (BCFB) and balanced performance control (BPC). To balance the interconnection and conflict between performance and saturation constraints, define a performance safety evaluation (PSE) function, which evaluates the system safety under the destabilizing effect variables (DEVs) like saturation quantity and filter errors, then the cumulative effects of  DEVs are fully utilized and compensated for the performance recovery. Specifically, there exists some degree of tolerance for the DEVs in the safety region, and the compensation operation works when the evaluation of the system goes dangerous. The advantages of the proposed methodology are illustrated in the numerical simulation. 
\end{abstract}


\IEEEpeerreviewmaketitle

\section{Introduction}

\IEEEPARstart{I}{n} practice,  there exit various constraints for nonlinear system control, such as (i) artificially added performance and safety constraints, (ii) capacity constraints due to physical stoppages, saturation. However, these two constraints are sometimes difficult to satisfy at the same time, for example, if a strict control performance is required but the system control capability is limited. 

In this paper we present a systematic approach to the balance between the performance and saturation constraints.

\emph{Performance-constrained control:}
Traditional efforts have to be made to the performance constraints in the fashion of prescribed performance control (PPC) and barrier Lyapunov function (BLF). The most typical and powerful tool for dealing with output and full (or partial) state constraints is the so-called barrier Lyapunov function (BLF) method \cite{tee2009barrier,tee2011control,tee2011control2}. With the help of the barrier function which grows to infinity at some finite limit, the state constraints are never violated. The effective prescribed performance control (PPC) \cite{bechlioulis2008robust,bechlioulis2011robust,
bechlioulis2016trajectory,Ma2021Adaptive} incorporates with equivalent design and barrier Lyapunov function, such that the system states are guaranteed to evolve within the user-specified prescribed performance functions (PPF), which determines convergence rate, overshoot and steady-state error, through error transformation function (ETF). 

\emph{Saturation-constrained control:}
It is clear that input saturation, as a type of the most important non-smooth nonlinearities in the industrial process, can frequently degrade the system performance, and even undermine the system stability. Accordingly, great attention has been paid to the control of input saturated systems, like 
anti-windup approach \cite{hu2008anti,grimm2004robust,Wang2020Output} and auxiliary system method \cite{chen2011adaptive,
he2015vibration}, etc.
For the PPC under input saturation, 
in \cite{wang2018online}, a constrained dynamic surface control with additional auxiliary system for saturation compensation was studied for the online performance-based control; in \cite{wang2019improved}, a novel prescribed performance controller with composite performance functions was constructed to circumvent the chattering of  control input with an auxiliary system avoiding the overrun of control input.


\emph{Balance between performance and saturation:}
In the existing considerable research results, some efforts have paid attention to the balance between the performance and saturation constraints. 
Model predictive controllers (MPC) can effectively handle nonlinearities and
can handle constraints on both the input and the states of a system \cite{HLi2014event,Johannes2019nonlinear}. Every control problem
can be formed as a constrained optimization problem in MPC such that the control performance can be optimized without saturation happens.
A flexible performance control (FPC) scheme \cite{yong2020flexible,Ji2021Saturation} introduces the modification signals into the predetermined performance function, thus features the capability of avoiding performance violation due to input saturation. 
 As the method termed explicit reference governor (ERG) for tacking problem developed in \cite{Nicotra19,Nicotra18,PWang2021prescribed}, the performance and control input constraints are converted into limitations on upper Lyapunov function (or the limitation in an invariant set), and the balance is achieved through the application of a trade-off reference trajectory. 

\emph{Contributions of this work:} 
Notably, seldom research takes into account that the limited saturation sometimes has a tolerable impact on the system performance, or the occurrence of saturation may sometimes positively accelerate system convergence.
Motivated by the above discussion, a systematic balanced control methodology is proposed for achieving the performance requirements while avoiding saturation from destroying system stability.

The main contributions are stated as follows:

\begin{enumerate} [(i)]
\item  Define a performance safety evaluation (PSE) function, which evaluates the system safety under the destabilizing effect variables (DEVs) like saturation quantity and filter errors. The previous balance operation between performance and saturation usually ignores the system's own adjustment capabilities, i.e., (a) the
constrained optimization problem in MPC rarely provides analytic solutions, and usually the solution demands a significant computation power; (b) any saturation impact in FPC is transmitted to the performance function, no matter if the saturation has a tolerable impact on the system performance; (c) the applied control constraints in MPC and ERG are hard limits in the invariant set, ignoring the fullest potential by saturating the inputs on purpose.   
In contrast, the PSE focuses on whether the system is in the safety region, not whether the control input is saturated.

\item Introduce a methodology for solving the balanced control between performance and saturation constraints with application to the named balanced command filtered backstepping (BCFB) and balanced performance control (BPC).  
The cumulative effects of the DEVs are fully utilized and compensated in the first-order dynamics of auxiliary systems of BCFB / dynamics of performance functions of BPC.  

\item Specifically,  according to the evaluation of PSE, (a) in BCFB, it is decided whether to unify the two errors (tracking error $s_1$ and compensated tracking error $z_1$) to the same within the safety region which is contained in the performance-constrained invariant set, or to compensate for DEVs such that the system remains stable; (b) in BPC, it is decided whether to  adhere to the performance function to converge at the established speed, or to compensate the performance function for DEVs to avoid constraint conflicts through the nature of the invariant set.
\end{enumerate}

\emph{Notation}:
Throughout this paper, ${\mathbb{R}}$ denotes the set of real numbers, ${\mathbb{R}}^{n\times m}$ represents the set of $n \times m$ real matrix. $\mathbb{I}_{i:j}$ with integers $i<j$ denotes the set $\lbrace i,i+1,\cdots,j\rbrace$.
$( \cdot )^T$ is the transpose of one matrix, 
diag$( \cdot )$ (or blk diag) represents the diagonal (block-diagonal) matrix. $\left|  \cdot  \right|$ is the absolute value of a scalar and $\left\|  \cdot  \right\|$ is defined as the Euclidean norm of a vector. Given a matrix $M$, $M>0$ (or $M<0$) means that $M$ is a positive definite (or negative definite) matrix.

\section{Problem Formulation and  Preliminaries}
Consider the following nonlinear system 
\begin{equation}\label{Eq_2_1}
\begin{aligned}
  & {{{\dot{x}}}_{i}}(t)={{f}_{i}}({{{\bar{x}}}_{i}})+{{g}_{i}}({{{\bar{x}}}_{i}}){{x}_{i+1}}+{{\omega }_{i}}(t),\ \ \ i\in {{\mathbb{I}}_{1:n-1}} \\ 
 & {{{\dot{x}}}_{n}}(t)={{f}_{n}}({{{\bar{x}}}_{n}})+{{g}_{n}}({{{\bar{x}}}_{n}})\text{sat}(u)+{{\omega }_{n}}(t), \\ 
 & y(t)={{x}_{1}}(t) \\ 
\end{aligned}
\end{equation}
where ${{x}_{i}}, {{y}}, {{u}_{i}}, {{\omega }_{i}}\in {\mathbb{R} }$ are the state, output, control input and external disturbance, respectively, ${{\bar{x}}_{i}}={{\left[ x_{1} ,x_{2} ,\ldots, x_{i}  \right]}^{T}}$ is the state vector of the first $i$ differential equations, ${{{x}}}={{\bar{x}}_{n}}$,
${{{\omega}}}={{\left[ \omega_{1} ,\omega_{2} ,\ldots, \omega_{n} \right]}^{T}}$ is the vector of disturbances. The functions ${{f}_{i}}({{\bar{x}}_{i}})$ and ${{g}_{i}}({{\bar{x}}_{i}})$, $ i\in {{\mathbb{I}}_{1:n}}$ are assumed to be known. The saturation function $\text{sat}(u)$ is given by
\begin{equation}\label{Eq_2_2}
\text{sat}({{u}})=\left\{ \begin{matrix}
\begin{aligned}
   &{{u}_{\max }},\ \ \  \text{if}\ {{u}}\ge{{u}_{\max }}  \\
   &{{u}},\ \ \ \ \ \ \ \ \text{if}\ {{u}_{\min }}<{{u}}<{{u}_{\max }}  \\
   &{{u}_{\min }}, \ \ \ \  \text{if}\ {{u}}\le{{u}_{\min }}  \\
\end{aligned}
\end{matrix} \right. 
\end{equation}
where ${{u}_{\min }}<0$ and ${{u}_{\max }}>0$ are known constants.

\textcolor{black}{
\emph{Control objective:} Design a control scheme for the nonlinear system (\ref{Eq_2_1}) such that 
\begin{enumerate}[(i)]
\item Saturation-constrained tracking control: In the presence of input saturation and external disturbance, the system output $y(t)$ tracks the desired trajectory ${{y}_{d}}(t)$  with the bounded tracking error ${{e}}(t)=y(t)-{{y}_{d}}(t)$;
\item Performance-constrained tracking control:  Despite affected by these effects, the tracking error ${{e}}(t)$ is constantly contained in the specified region $\left\{ e | -\underline{\delta }\rho (t)<e(t)<\bar{\delta }\rho (t) \right\}$ where $\underline{\delta }, \bar{\delta }>0 $, ${{\rho }} (t)$    represents the performance constraint function which ultimately approaches the steady bound ${{\rho }_{ \infty }}$. 
\end{enumerate}}

To ensure controllability, we will invoke the following assumption, which is standard in command filtered backstepping.

\begin{assumption} \label{Assumption_1}
\cite{farrell2009command} Let the domain $D \subset \mathbb{R}^n $  denote the compact set that contains the origin, the initial condition $x(0)$ and the trajectory $y_d(t)$. For system (\ref{Eq_2_1}), as  $i\in {{\mathbb{I}}_{1:n}}$, 
 \begin{enumerate}[(i)]
\item  $f_i({{{\bar{x}}}_{i}})$ and $g_i({{{\bar{x}}}_{i}})$ and their first partial derivatives are continuous and bounded on $D$ for all $t \ge 0$. 
Each $f_i({{{\bar{x}}}_{i}})$ is locally Lipschitz in ${{{\bar{x}}}_{i}}$;  
\item there exist positive constants $g_{i,m}$ and $g_{i,M}$ such that $g_{i,M} > \vert g_i({{{\bar{x}}}_{i}}) \vert > g_{i,m}$ on $D $ for all $t \ge 0$. This implies that each $g_i({{{\bar{x}}}_{i}})$ has a constant, known sign. 
\item the external disturbance $\omega_i(t)$ is continuous and bounded.
\end{enumerate}
\end{assumption}



\begin{assumption} \label{Assumption_2}
For $t \ge 0$, the target signal $y_d(t)$ and the performance function $\rho(t)$ and their corresponding partial derivatives are continuous, bounded and available. Moreover, the initial condition  $e(0)$ satisfies $-\underline{\delta }\rho (0)<e(0)<\bar{\delta }\rho (0) $.
\end{assumption}

\begin{assumption} \label{Assumption_3}
\cite{yong2020flexible} Under the given initial conditions, there exists a feasible control strategy   $u(t)$ to achieve the tracking control objective for a practical plant in the form of the system (\ref{Eq_2_1}) subject to the input saturation described by (\ref{Eq_2_2}) and bounded external disturbances.
\end{assumption}

\begin{remark} \label{remark_feasibility}
Notably, an open loop unstable system like (\ref{Eq_2_1}) can not be globally stabilized in the presence of the input saturation \cite{CWen11robust}.
The input saturation affects the feasibility of a control problem  mainly through the system dynamics  (the system nonlinearities). At the best case, the local stability results
are expected which depends on the growth properties of the system nonlinearities; or at the worst case, the  instability results do not admit any control scheme for tracking control.  For instance, consider the system $\dot{x}=2+x^2+(1+0.5\sin(x))\text{sat}(u)$ where the saturation function takes values within $[-1,1]$. No matter what the
input signal $u$ is, the state $x$ will tend to infinity even in finite time.  Moreover, many
practical plants always operate in some specified regions where they are controllable under the input saturation, as a result, Assumption \ref{Assumption_3} is reasonable.
\end{remark}

Below we provide two propositions regarding to the invariant set for the performance-constrained control.

%

\begin{proposition}   \label{Proposition_Tracking}
Consider the system in form of
\begin{equation} \label{Eq_2_3}
\dot{z}=A({{\bar{x}}_{n}})z+\omega
\end{equation}
where $z \in \mathbb{R}^n$ is the state vector, the system matrix $A({{\bar{x}}_{n}})$ is given by
$A({{\bar{x}}_{n}})={{A}_{0}}+{{A}_{g}}({{\bar{x}}_{n}})$
with the constant diagonal matrix ${{A}_{0}}=\text{diag}(-{{k}_{1}},\ldots ,-{{k}_{n}})$ ($k_i>0$, $i\in {{\mathbb{I}}_{1:n}}$) and the skew-symmetric matrix
\[{{A}_{g}}({{\bar{x}}_{n}})=\left[ \begin{matrix}
   0 & {{g}_{1}}({{x}_{1}}) & {} & {}  \\
   -{{g}_{1}}({{x}_{1}}) & 0 & {{g}_{2}}({{{\bar{x}}}_{2}}) & {}  \\
   {} & \ddots  & \ddots  & \ddots   \\
   {} & {} & -{{g}_{n-1}}({{{\bar{x}}}_{n-1}}) & 0  \\
\end{matrix} \right].\]
Define a Lyapunov function candidate $V={{z}^{T}}Pz$ with $P$ being a symmetric positive-definite matrix.
Given the performance function 
\begin{equation}  \label{Eq_2_6}
\rho (t)=({{\rho }_{0}}-{{\rho }_{\infty }}){\text{exp}({-\kappa t}})+{{\rho }_{\infty }}
\end{equation}
with $\rho_0=\rho(0)>\rho_\infty>0$, $\kappa>0$, there exists a performance constraint 
$\vert z_1\vert=\vert C z\vert \le \rho(t)$  with $C={{\left[ 1,0,\ldots 0 \right]}}$.
Let 
\begin{enumerate}[(i)]
\item  the following matrix inequalities are feasible  
\begin{equation} \label{Eq_2_4}
\begin{aligned}
& {{A}_{0}}X+XA_{0}^{T}+(\epsilon +\alpha +2\kappa)X+\alpha W\le 0, \\
& {{A}_{g}}X+XA_{g}^{T}- \epsilon X \le 0
\end{aligned}
\end{equation}
for $X={{X}^{T}}={{P}^{-1}}$, $W>0$, $\alpha>0$ and $\epsilon>0$;
\item  the function $\omega (t)$ satisfies the inequality ${{\omega }^{T}}W\omega \le 1 \le \Gamma_\infty$ where 
$\Gamma_\infty= \rho_{\infty}^2/(CXC^T) $;
\end{enumerate}
then the region of $\Omega=\left\{ z|V(t)\le \Gamma(t) \right\}$ with $\Gamma(t)= \rho(t)^2/(CXC^T) $ is an invariant set in which the performance constraint is guaranteed.
\end{proposition}

\begin{proof}
See Appendix A.
\end{proof}

\begin{remark}
Looking into the influence of $\rho_\infty$ and $W$ on the size of $P$, one obtains
\begin{enumerate} [(i)]
\item the more stringent performance requirements (smaller $\rho_\infty$) bring about a larger $P$;
\item on the contrary, the larger external disturbance (smaller matrix $W$) brings about a smaller $P$.
\end{enumerate} 
Therefore, the choice of $\rho_\infty$ needs to be related to the size of the disturbance. Accordingly, if the existing disturbance is small enough to be ignored or the disturbance observer is perfect enough to eliminate the disturbance, $W$ can be regarded large enough and $\rho_\infty$ can be taken as small as possible.
\end{remark}

\begin{remark}
For the satisfaction of matrix inequalities in (\ref{Eq_2_4}), a trivial solution is to let $X=V_h I_n$ where $V_h$ is the value to be optimized and $I_n$ is a $n$-dimensional unit matrix such that the second inequality in (\ref{Eq_2_4}) is obviously established.
\end{remark}

\begin{proposition} \label{Proposition_Performance}
Consider the system
\begin{equation}  \label{Eq_2_5} 
\dot{z}=A({{\bar{x}}_{n}},\mu )z+{{D}_{\mu }}\omega. 
\end{equation}
where $z \in \mathbb{R}^n$ is the state vector, $\mu$ is a design parameter, the matrix ${{D}_{\mu }}=D_{\mu }^{1}$ is in form $D_{\mu }^{\alpha }=\text{diag}({{\mu }^{\alpha }},1,\ldots ,1)$ and the system matrix $A({{\bar{x}}_{n}},\mu )$ is given by
$A({{\bar{x}}_{n}},\mu )=D_{\mu }^{\nu }{{A}_{0}}+{{D}_{\mu }}{{A}_{g}}({{\bar{x}}_{n}}){{D}_{\mu }}$, $\nu>-1$.
Define a Lyapunov function candidate $V={{z}^{T}}Pz$ where $P=\text{blk}\ \text{diag}({{P}_{1}},{{P}_{2-n}})$ with ${{P}_{1}} \in \mathbb{R}$ and ${{P}_{2-n}} \in \mathbb{R}^{(n-1)\times (n-1)}$ being symmetric positive-definite matrices. Let 
\begin{enumerate} [(i)]
\item the matrix inequalities in (\ref{Eq_2_4}) are feasible 
for $X={{X}^{T}}$, $X=\text{block}\ \text{diag}({{X}_{1}},{{X}_{2-n}})$, $P={{X}^{-1}}$ (both $X_1$ and ${{X}_{2-n}}$ are symmetric definite matrices), $W>0$, $\alpha>0$ and arbitrary small $\epsilon>0$;
\item the disturbance function $\omega (t)$ satisfies the inequality ${{\omega }^{T}}D_{\mu }^{1-\nu/2 }WD_{\mu }^{1-\nu/2 }\omega \le 1$, 
\end{enumerate}
then the region of $\Omega=\left\{ z|\Phi={{z}^{T}}D_{\mu }^{{\nu }/{2}\;}PD_{\mu }^{{\nu }/{2}\;}z\le 1 \right\}$ is an invariant set.
\end{proposition}
 
\begin{proof}
See Appendix B.
\end{proof} 


\section{Balanced command filtered backstepping}

\subsection{Design equations}
In the balanced command filtered backstepping (BCFB) approach, we view the state variable ${{x}_{i+1}}$ as the control input to the ${{x}_{i}}$-subsystem and design a stabilizing ${{x}_{i,d}}$ which would achieve the control objective.  At step $i\in {{\mathbb{I}}_{2:n}} $ of the backstepping procedure, a command filter will be used to get the differential estimation of a known signal without analytic computation, which is defined by
\begin{equation}  
{{\tau }_{i}}{{\dot{x}}_{i,c}}+{{x}_{i,c}}={{x}_{i,d}}, \ \ x_{i,c}(0)=x_{i,d}(0),\ \ i\in {{\mathbb{I}}_{2:n}}. \\ 
\end{equation}
where $\tau_i$  is a positive design parameter, $x_{i,c}$ represents the state of
the command filter.  

To compensate the filtering errors ${{x}_{i,e}}=x_{i,c}-x_{i,d},\ \ i\in {{\mathbb{I}}_{2:n}} $ and the influence of saturation $\Delta u=\text{sat}(u)-u$, the following auxiliary system is constructed  
\begin{equation}\label{Eq_3_2_2}
\begin{aligned}
  & {{{\dot{\eta }}}_{1}}=-f_p(z){{k}_{1}}{{\eta }_{1}}+(1-f_p(z)){{g}_{1}}({{x}_{1}}) \left( {{\eta }_{2}}+{{x}_{2,e}}\right), \\ 
 & {{{\dot{\eta }}}_{i}}=-{{k}_{i}}{{\eta }_{i}}+{{g}_{i}}({{{\bar{x}}}_{i}})({{\eta }_{i+1}}+{{x}_{i+1,e}}),\ \ i\in {{\mathbb{I}}_{2:n-1}} \\ 
 & {{{\dot{\eta }}}_{n}}=-{{k}_{n}}{{\eta }_{n}}+{{g}_{n}}({{{\bar{x}}}_{n}})\Delta u  
\end{aligned}
\end{equation}
where  ${{{\eta}}}={{\left[\eta_{1},\eta_{2},\ldots, \eta_n \right]}^{T}}$, the initial condition is $\eta_i(0)=0$, $i\in {{\mathbb{I}}_{1:n}}$ and $k_i>0, {{\mathbb{I}}_{1:n}}$ are designed control parameters. The $f_p(z) \in [0,1]$ is a performance safety evaluation (PSE) function  whose detailed expression will be presented later.

On this basis, define the tracking error ${{s}_{i}}$ and the compensated tracking error  ${{z}_{i}}$ as 
\begin{equation} \label{Eq_3_2_1}
\begin{aligned}
&{{s}_{1}}={{x}_{1}}-{{y}_{d}}, 
 \ \ \ \ \ \ \ \   {{s}_{i}}={{x}_{i}}-{{x}_{i,c}},\ \ \ \ \ \ \ \ i\in {{\mathbb{I}}_{2:n}}  \\
 & {{z}_{1}}={{x}_{1}}-{{y}_{d}}-{{\eta }_{1}},   
  \ \ {{z}_{i}}={{x}_{i}}-{{x}_{i,c}}-{{\eta }_{i}},\ \ i\in {{\mathbb{I}}_{2:n}}. 
\end{aligned} 
\end{equation} 
The notations given  by  ${{{s}}}={{\left[s_{1},s_{2},\ldots, s_n \right]}^{T}}$,  ${{{z}}}={{\left[z_{1},z_{2},\ldots, z_{n} \right]}^{T}}$ are defined.

One can design the stabilizing function  
\begin{align} \label{Eq_3_2_4}
  & {{x}_{2,d}}=g_{1}^{{-1}}({{x}_{1}})\left( -{{k}_{1}}{{z}_{1}}-{{f}_{1}}({{x}_{1}})+{{{\dot{y}}}_{d}} \right), \\ 
 & {{x}_{i+1,d}}=g_{i}^{{-1}}({{{\bar{x}}}_{i}})\left( -{{k}_{i}}{{s}_{i}}-{{f}_{i}}({{{\bar{x}}}_{i}})-{{g}_{i-1}}({{{\bar{x}}}_{i-1}}){{z}_{i-1}}+{{{\dot{x}}}_{i,c}} \right),   \notag \\ 
 & u=g_{n}^{{-1}}({{{\bar{x}}}_{n}})\left( -{{k}_{n}}{{s}_{n}}-{{f}_{n}}({{{\bar{x}}}_{n}})-{{g}_{n-1}}({{{\bar{x}}}_{n-1}}){{z}_{n-1}}+{{{\dot{x}}}_{n,c}} \right) \notag 
\end{align}
where $ i\in {{\mathbb{I}}_{2:n-1}}$. Therefore, through the substitution of (\ref{Eq_3_2_4}) into (\ref{Eq_2_1}), the system model gets evolved into 
\begin{align} \label{Eq_3_2_5}
  & {{{\dot{z}}}_{1}}=-{{k}_{1}}{{z}_{1}}+{{g}_{1}}({{x}_{1}}){{z}_{2}}  \nonumber \\
  & \ \ \ \ \ \ +f_p(z)\left( {{k}_{1}}{{\eta }_{1}}+{{g}_{1}}({{x}_{1}})({{\eta }_{2}}+{{x}_{2,e}} )\right)+{{\omega }_{1}}, \nonumber  \\ 
 & {{{\dot{z}}}_{i}}=-{{k}_{i}}{{z}_{i}}+{{g}_{i}}({{{\bar{x}}}_{i}}){{z}_{i+1}}-{{g}_{i-1}}({{{\bar{x}}}_{i-1}}){{z}_{i-1}}+{{\omega }_{i}},  i\in {{\mathbb{I}}_{2:n-1}}\ \nonumber  \\ 
 & {{{\dot{z}}}_{n}}=-{{k}_{n}}{{z}_{n}}-{{g}_{n-1}}({{{\bar{x}}}_{n-1}}){{z}_{n-1}}+{{\omega }_{n}}  
\end{align}
which is further rewritten into the following compact form
\begin{equation}\label{Eq_3_2_6}
\dot{z}=A({{\bar{x}}_{n}})z+f_p(z){{B}_{1}}\left( {{k}_{1}}{{\eta }_{1}}+{{g}_{1}}({{x}_{1}})({{\eta }_{2}}+{{x}_{2,e}}) \right)+\omega 
\end{equation} 
where $B_1={{\left[ 1,0,\ldots 0 \right]}^{T}}$ and the system matrix $A({{\bar{x}}_{n}})$ is given in Proposition \ref{Proposition_Tracking}.  


\textcolor{black}{In the following, the balance between performance and saturation is to be discussed with the performance constraint defined as }

\emph{Performance constraint:}
with respect to the tracking performance of $z_1$, similar to the prescribed performance control, define the prescribed performance function (PPF) given in
(\ref{Eq_2_6}). The constant $\kappa$ represents the decreasing rate, a lower bound on the required speed of convergence.  The performance constraint is thus described by $\vert z_1\vert=\vert C z\vert \le \rho(t)$. 
Furthermore, the convergence rate $\kappa$  is subject to $0<\kappa \le k_1$ since the ideal convergence rate of $z_1$ equals to $k_1$ in the  situation $\dot{z}_1=-k_1z_1$.

Here two compact sets are defined and their relationship is illustrated in Fig. \ref{fig:Fig_explain}.

\begin{figure}[!t]
\centerline{\includegraphics[width=0.33\textwidth]{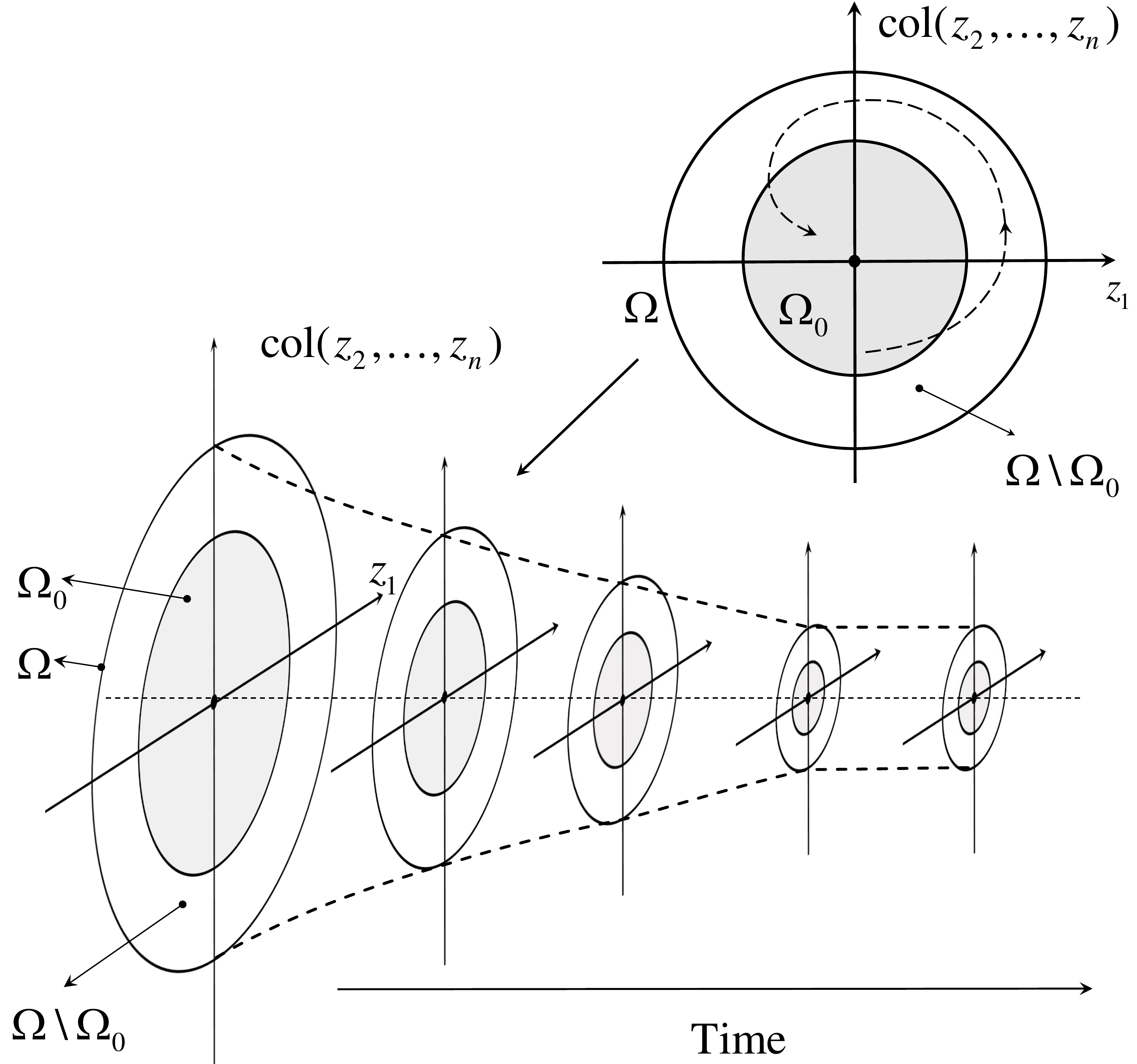}}
\caption{Illustration of invariant set $\Omega$, safety region $\Omega_0$.}
\label{fig:Fig_explain}
\end{figure} 

\emph{User-defined invariant set}: define the region $\Omega =\left\{ z|V(t)={{z}^{T}}Pz\le \Gamma(t)\right\}$, $P>0$, as an invariant set such that once a trajectory of the system enters $\Omega $, it will never leave it again. Meanwhile, the performance constraint $\vert C z\vert \le \rho(t)$ gets satisfied through the time-varying $\Gamma(t)$ whose detailed expression can refer to Proposition \ref{Proposition_Tracking}.   

\emph{User-defined safety region}: define the region ${{\Omega }_{0}}=\left\{ z|V(t)\le {{\Gamma}_{0}}(t) =\sigma \Gamma(t) \right\}$, $0<\sigma<1$, as a safety region within which the system performance is evaluated safe, and the destabilizing effect from the DEVs, like $x_{i,e}$, $i\in {{\mathbb{I}}_{2:n}}$ and $\Delta u$, are affordable. The user-defined variable $\sigma$ determines the safety region $\Omega_0$'s size.

The specific impact of DEVs on performance can be viewed from the following two perspectives
\begin{enumerate} [(i)]
\item An ideal situation is that even if saturation occurs, the system has the ability to self-regulate and will not let the error state run out of the region $\Omega_0$. Therefore, there is no need to deal with input saturation in this case, corresponding to $f_p=1$; 
\item For some systems with insufficient self-regulation ability, the input saturation makes the error state escape the region $\Omega_0$, which further causes system divergence. Therefore, some mandatory processing is required for the escaping behavior, corresponding to $f_p=0$.
\end{enumerate}
Consequently, the balance between the performance and saturation constraints can be directly achieved through the choice of PSE $f_p(z)$ given as
\begin{equation} \label{Eq_3_2_6b}
f_{p}(z)=\left\{ \begin{matrix}
\begin{aligned}
&  0,\ \ \ \ \ \ \ \ \ \ \ \ \ \ \ \ \ \ \ \ \ \ \  \text{if}\ V \ge \Gamma  \\
&  {(\Gamma -V)}/{(\Gamma -{{\Gamma}_{0}})},\ \   \text{if}\ \Gamma_{0}<V<\Gamma  \\
&  1,\ \ \ \ \ \ \ \ \ \ \ \ \ \ \ \ \ \ \ \ \ \ \  \text{if}\ V\le \Gamma_{0} \\
\end{aligned}
\end{matrix} \right.
\end{equation}

\subsection{Stability analysis}

\begin{theorem} \label{Theorem_BCF}
Consider the nonlinear system in form of (\ref{Eq_2_1}) under Assumptions \ref{Assumption_1}-\ref{Assumption_3}. With the auxiliary system (\ref{Eq_3_2_2}), the control law is designed according to (\ref{Eq_3_2_4}). The PPF is given in (\ref{Eq_2_6}). The
PSE $f_p(z)$ is specified as  (\ref{Eq_3_2_6b}).  Then, for the initial condition $s(0)=z(0) \in \Omega$, we have the following conditions: (i) the compensated tracking error $z$ is always contained in the invariant set $\Omega$ and satisfies the performance constraint $\vert z_1(t) \vert \le \rho(t)$; (ii) in the safety region $z \in \Omega_0$, the tracking error $s_1$ and the compensated tracking error $z_1$ are close to the same, despite of the DEVs; (iii) the tracking error $s_1$ converges into a sufficiently small neighborhood of the origin and all signals in the closed-loop system are bounded.
\end{theorem} 
\begin{proof}
If, as we said in Theorem \ref{Theorem_BCF} (ii), $s_1$ and $z_1$ tend to be the same in $\Omega_0$, then we can focus on the convergence of $z$ beyond $\Omega_0$ and the approximate convergence of $s_1$ (approximately $z_1$)  when $z$ returns within $\Omega_0$. As a consequence, choose the Lyapunov function candidate as $V_s=0.5 z_1^2+0.5\sum\nolimits_{i=2}^{n}{(z_{i}^{2}+\eta_{i}^{2})}$ whose time derivative yields
\begin{equation}
\dot{V}_s\le -k_s V_s +\sum\nolimits_{i=1}^{n} { x_{i,e}^{2} }+\delta _{\omega }^{2}+\Delta u^2
\end{equation}
 where $k_s =\min \left\{ 2\bar{k}_i,i\in {{\mathbb{I}}_{1:n}}  \right\}$, $\bar{k}_{1}=(1-0.5f_p)k_1-0.5 f_p g_1^2-0.25$, $\bar{k}_i=k_i-0.5g_i^2-1$, $ i\in {{\mathbb{I}}_{2:n-1}} $, $\bar{k}_n=k_n-0.25g_n^2-1$ and  $ \Vert \omega(t) \Vert \le \delta_{\omega}^2$.  
The stability of $V_s$ can only be obtained by guaranteeing the signum of $k_s$ and exploiting the boundedness of filter errors ${{x}_{i,e}}$, $ i\in {{\mathbb{I}}_{2:n}} $ and saturation quantity $\Delta u$. Based  on the proof of Theorem 2 in \cite{farrell2009command}, the boundedness of filter error ${{x}_{i,e}}$ can be obtained since the function $x_{i,d}$, its first partial derivatives with respect to $(s_i,\eta_i,x_{i,c})$ and its time derivative $\dot{x}_{i,d}$ are continuous and bounded under 
Assumption \ref{Assumption_1}-\ref{Assumption_2}. With regard to the boundedness of saturation quantity $\Delta u$, a commonly used approach is to show that the various quantities participating in $u(t)$ are bounded, so the saturation $\Delta u$ does not make a large quantity. 

The remaining proof is divided into two parts: (i) the invariance of $z$ in the set $\Omega$; (ii) the tracking performance in the safety region $\Omega_0$ including the consistency of $s_1$ and $z_1$.

\emph{Case 1.} $z \notin \Omega $ which brings $f_{p}=0$. In such case, the evaluation of the system is in a dangerous state, that is, not in the user defined invariant set and the system model (\ref{Eq_3_2_6}) is converted into a simplified form
$\dot{z}=A({{\bar{x}}_{n}})z+\omega$. According to Proposition \ref{Proposition_Tracking}, the invariance of $z$ in the set $\Omega$ is obvious established. Furthermore, it is easy to  make $k_s>0$ due to $f_p=0$ such that the stability of $V_s$ is achieved.

\emph{Case 2.} $z \in \Omega_0 $ which brings $f_p=1$. The corresponding composite system dynamics are obtained as
\begin{equation}\label{Eq_3_2_10}
\begin{aligned}
  & \dot{z}=A({{{\bar{x}}}_{n}})z+{{B}_{1}}\left( {{k}_{1}}{{\eta }_{1}}+{{g}_{1}}({{x}_{1}})({{\eta }_{2}}+ {{x}_{2,e}}) \right)+\omega,  \\ 
 & {{{\dot{\eta }}}_{1}}=-{{k}_{1}}{{\eta }_{1}}. \\ 
\end{aligned}
\end{equation}
It is clear that ${{\eta }_{1}}$ converges to zeros over time and ${{z}_{1}}={{s}_{1}}=y-{{y}_{d}}$, which is the control target we really want. 
\begin{enumerate} [(i)]
\item  Under the premise of ${{\Omega }_{0}}$,  when saturation does not occur, we can get that ${{\eta }_{i}}$, $i\in {{\mathbb{I}}_{2:n}}$ are bounded and of small quantity, and the system is fully robust to these destabilizing effects of filter errors. 

\item  Consider the existence of input saturation in the safety region. 
One optimistic situation is  a certain degree of saturation would not have or have little negative impact on performance, such that the system remains in the safety region.
The opposite is that the saturation would drive the trajectory of the system out of the safety region ${{\Omega }_{0}}$, stay in the region  $\Omega \backslash {{\Omega }_{0}}$ and never run out of $\Omega $ due to the PSE $f_p$ decreasing to 0 as the trajectory close to the boundary of $\Omega $. In light of Assumption \ref{Assumption_3}, when the saturation disappears or the system's adjustment ability has tolerance to saturation, the system trajectory will return to the safe area. 
\end{enumerate}

From the above analysis, it is thus concluded that the conflict between performance and saturation gets balanced. 
\end{proof}

\begin{remark} \label{Remark_CFB}
In the command filtered backstepping (CFB), the auxiliary system is governed by \begin{align}  \label{Eq_3_2_11}
  & {{{\dot{\eta }}}_{i}}=-{{k}_{i}}{{\eta }_{i}}+{{g}_{i}}({{{\bar{x}}}_{i}})({{\eta }_{i+1}}+  {{x}_{i+1,e}}),\ \ i\in {{\mathbb{I}}_{1:n-1}}  \nonumber \\ 
 & {{{\dot{\eta }}}_{n}}=-{{k}_{n}}{{\eta }_{n}}+{{g}_{n}}({{{\bar{x}}}_{n}})\Delta u  
\end{align}  
and the corresponding virtual controller ${{x}_{i+1,d}}$,  $ i\in {{\mathbb{I}}_{2:n-1}}$,  and controller $u$ are the same to (\ref{Eq_3_2_4}) but
\begin{align}  \label{Eq_3_2_12}
{{x}_{2,d}}=g_{1}^{{-1}}({{x}_{1}})\left( -{{k}_{1}}{{s}_{1}}-{{f}_{1}}({{x}_{1}})+{{{\dot{y}}}_{r}} \right).
\end{align}
The auxiliary signal ${{\eta }_{i}}$ actually serves as a collection and transmission of negative effects from filter error ${{x}_{i,e}}$ and saturation $\Delta u$. 
The subtraction between ${{s}_{i}}$ and ${{\eta }_{i}}$ in the supplementary error ${{z}_{i}}$  is actually to delete these effects. 
\end{remark}

\begin{remark}
Here we compare the proposed balanced command filtered backstepping control scheme with the existing CFB methods from the following aspects: 
\begin{enumerate} [(i)]
\item To our best knowledge, this is the first time that the performance constraint on the command filtered backstepping method has been proposed without the operation of error transformation like PPC;
\item In general, the PSE $f_p(z)$ determines the role played by $\eta_1$ which collects the cumulative impact of DEVs: $f_p(z)=1$, $\eta_1$ is responsible for unifying the two errors  $s_1$ and $z_1$ to the same; $f_p(z)=0$, $\eta_1$ is responsible for compensating DEVs such that the tracking error $s(t)$ returns to the invariant set $\Omega$ of performance constraints. 
\end{enumerate}
\end{remark}

\section{Balanced performance control}
Inspired by PPC, if the performance constraint is imposed on the tracking error $e(t)$, the performance-constrained error gets converted into an equivalent unconstrained variable through the error transform function (ETF). 
Given the performance constraint described by  
$-\underline{\delta }\rho (t)<e(t)<\bar{\delta }\rho (t)$,
following the ETF, we have
$e(t)=\rho (t)T({{z}_{1}})$
where ${{z}_{1}}(t)$ denotes the transformed error and $T({{z}_{1}})$ is a smooth, strictly monotonic increasing function. 
If ${{z}_{1}}(t)$ is bounded, the performance constraint of $e(t)$ is obvious satisfied.

The specific ETF $T(z_1)$ can be expressed as
	\[T(z_1)=\frac{e}{\rho}=\frac{\bar{\delta }\exp ({{z}_{1}})-\underline{\delta }\exp (-{{z}_{1}})}{\exp ({{z}_{1}})+\exp (-{{z}_{1}})}.\] 	 	  
Differentiating $z_1$ with respect to time, we have
\begin{equation}
\dot{z}_{1}=\mu \left( e,\rho  \right)\left(\dot{e}-\dot{\rho }({e}/{\rho })\right) 	
\end{equation}
with $\mu \left( e,\rho  \right)=\lambda \left( e,\rho  \right) /{\rho (t)}$ and 
$\lambda \left( e,\rho  \right)=\frac{\partial T_{\text{inv}}}{\partial \left( {e(t)}/{\rho (t)}\; \right)} $ where  $T_{\text{inv}}(\cdot)$ is the inverse function of $T(\cdot)$.

\begin{remark}
In a prescribed performance function, a larger decreasing rate of $\rho(t)$ increases the possibility of a singularity problem. In particular, if the tracking error $e(t)$ may get closer to the constraint performance, i.e., ${e(t)}/{\rho (t)}\;\to \bar{\delta },-\underline{\delta }$, then $\lambda \left( e,\rho  \right)\to \pm \infty $, which gives rise to the input saturation problem and to a violation of the prescribed constraint conditions.
\end{remark}

\subsection{Design equations}
Following the framework of command filtered backstepping, let us define the new error coordinates
\begin{equation} \label{Eq_4_3}
\begin{aligned}
  & {{z}_{1}}={{s}_{1}}={{T}^{-1}}\left( {e}/{\rho }\; \right),  \ \ \ {{s}_{i}}={{x}_{i}}-{{x}_{i,c}},\ \ i\in {{\mathbb{I}}_{2:n}}  \\ 
 & {{z}_{i}}={{x}_{i}}-{{x}_{i,c}}-{{\eta }_{i}},\ \ i\in {{\mathbb{I}}_{2:n}} \\ 
\end{aligned}
\end{equation}
and the auxiliary system for the compensation of DEVs is governed by
\begin{align} \label{Eq_4_4}
 & {{{\dot{\rho }}}}= -f_p(z,\mu){{k}_{\rho}}(\rho -{{\rho }_{\infty }}) \nonumber \\
 & \ \ \ \ \ +(1-f_p(z,\mu))f_t(e,\rho)({\rho}/{e}){{g}_{1}}({{x}_{1}})\left( {{\eta }_{2}}+{{x}_{2,e}} \right),   \nonumber \\ 
 & {{{\dot{\eta }}}_{i}}=-{{k}_{i}}{{\eta }_{i}}+{{g}_{i}}({{{\bar{x}}}_{i}}){{\eta }_{i+1}}+{{g}_{i}}({{{\bar{x}}}_{i}}){{x}_{i+1,e}},\ \ i\in {{\mathbb{I}}_{2:n-1}} \nonumber \\ 
 & {{{\dot{\eta }}}_{n}}=-{{k}_{n}}{{\eta }_{n}}+{{g}_{n}}({{{\bar{x}}}_{n}})\Delta u  
\end{align}
whose initial condition is $\rho(0)=\rho_0$, $\eta_i(0)=0$, $i\in {{\mathbb{I}}_{2:n}}$. The $f_p(z,\mu) \in [0,1]$ is a performance safety evaluation (PSE) function related to the error $z$ and the performance-related $\mu$.  The $f_t(e,\rho) \in [0,1]$ is a dead zone transition (DZT) function for the region around $e=0$. Their specific forms of $f_p(z,\mu)$ and $f_t(e,\rho)$ will be given later.

Through deduction, the derivative of ${{z}}$ becomes
\begin{align} \label{Eq_4_5}
  & {{{\dot{z}}}_{1}}=\mu \left( {{f}_{1}}({{x}_{1}})+{{g}_{1}}({{x}_{1}}) \right. \notag \\
  & \ \ \ \ \  \times \left( {{x}_{2,d}}+{{z}_{2}}+{{\eta }_{2}}+{{x}_{2,e}} \right) \left.-{{{\dot{y}}}_{r}}-\dot{\rho }({e}/{\rho })+{{\omega }_{1}} \right), \notag \\ 
 & {{{\dot{z}}}_{i}}={{f}_{i}}({{{\bar{x}}}_{i}})+{{g}_{i}}({{{\bar{x}}}_{i}})({{x}_{i+1,d}}+{{z}_{i+1}})   +{{k}_{i}}{{\eta }_{i}}-{{{\dot{x}}}_{i,c}}+{{\omega }_{i}}, \notag\\ 
 & {{{\dot{z}}}_{n}}={{f}_{n}}({{{\bar{x}}}_{n}})+{{g}_{n}}({{{\bar{x}}}_{n}})u+{{k}_{n}}{{\eta }_{n}}-{{{\dot{x}}}_{n,c}}+{{\omega }_{n}}   
\end{align}
where $i\in {{\mathbb{I}}_{2:n-1}}$. By viewing $x_{i+1,d}$ as a virtual control to $x_i$-dynamics, we obtain the stabilizing functions
\begin{align} \label{Eq_4_6}
  & {{x}_{2,d}}=g_{1}^{{-1}}({{x}_{1}})\left( -{{k}_{1}}{{\mu }^{\nu -1}}{{z}_{1}}-{{f}_{1}}({{x}_{1}})+{{{\dot{y}}}_{r}} \right), \\ 
 & {{x}_{3,d}}=g_{i}^{{-1}}({{{\bar{x}}}_{i}})\left( -{{k}_{2}}{{s}_{2}}-{{f}_{2}}({{{\bar{x}}}_{2}})-{{g}_{1}}({{x}_{1}})\mu {{z}_{1}}+{{{\dot{x}}}_{2,c}} \right), \notag \\ 
 & {{x}_{i+1,d}}=g_{i}^{{-1}}({{{\bar{x}}}_{i}})\left( -{{k}_{i}}{{s}_{i}}-{{f}_{i}}({{{\bar{x}}}_{i}})-{{g}_{i-1}}({{{\bar{x}}}_{i-1}}){{z}_{i-1}}+{{{\dot{x}}}_{i,c}} \right), \notag  \\ 
 & u=g_{n}^{{-1}}({{{\bar{x}}}_{n}})\left( -{{k}_{n}}{{s}_{n}}-{{f}_{n}}({{{\bar{x}}}_{n}})-{{g}_{n-1}}({{{\bar{x}}}_{n-1}}){{z}_{n-1}}+{{{\dot{x}}}_{n,c}} \right) \notag 
\end{align}
with $i\in {{\mathbb{I}}_{3:n-2}}$, the parameter $\nu >-1$. At the end of the recursive procedure, we obtain the compensated error system
\begin{align} \label{Eq_4_7}
  & {{{\dot{z}}}_{1}}=-{{k}_{1}}{{\mu }^{\nu }}{{z}_{1}}+\mu({{g}_{1}}({{x}_{1}}) {{z}_{2}}+ {{\omega }_{1}})+\mu f_p{{k}_{\rho}}({e}/{\rho })  (\rho -{{\rho }_{\infty }})  \nonumber \\ 
 & \ \ \ \ \ \ + \mu [1+(f_p-1)f_t]{{g}_{1}}({{x}_{1}})\left( {{\eta }_{2}}+{{x}_{2,e}} \right),  \nonumber  \\ 
 & {{{\dot{z}}}_{2}}=-{{k}_{2}}{{z}_{2}}+{{g}_{2}}({{{\bar{x}}}_{2}}){{z}_{3}}-\mu{{g}_{1}}({{x}_{1}}) {{z}_{1}}+{{\omega }_{2}},  \nonumber  \\ 
 & {{{\dot{z}}}_{i}}=-{{k}_{i}}{{z}_{i}}+{{g}_{i}}({{{\bar{x}}}_{i}}){{z}_{i+1}}-{{g}_{i-1}}({{{\bar{x}}}_{i-1}}){{z}_{i-1}}+{{\omega }_{i}},\ \ i\in {{\mathbb{I}}_{3:n-1}}\ \nonumber  \\ 
 & {{{\dot{z}}}_{n}}=-{{k}_{n}}{{z}_{n}}-{{g}_{n-1}}({{{\bar{x}}}_{n-1}}){{z}_{n-1}}+{{\omega }_{n}}.   
\end{align}

\begin{remark}
Influence of $\nu$'s selection. Invoking the definition of $\mu(e,\rho)$, what $\mu(e,\rho)$ represents is the partial derivative of the transformed error variable $z_1$ with respect to $e$.  
\begin{enumerate} [(i)]
\item If $\nu \le -1$, the performance constraint will be reduced to the control of $e$, making the performance constraint meaningless, so this is not desirable;
\item In addition, the choice of large $\nu$ would undoubtedly increase the stabilizing function $x_{i,d}$ of various orders, and the input control quantity $u$ would be more sensitive to external disturbances and noise.
\end{enumerate}
\end{remark}

\begin{remark}
No $\dot{\rho }({e}/{\rho })$'s appearance in ${{x}_{2,d}}$. Different from the traditional PPC, we give up the compensation for the performance derivative in ${{x}_{2,d}}$, but choose to treat this item as an interface for the impact of DEVs. 
This contributes to the dynamics of $\rho$ in (\ref{Eq_4_4}). 
\end{remark}

Let us rewrite the system (\ref{Eq_4_7}) in the following compact form
\begin{equation} \label{Eq_4_8}
\begin{aligned}
\dot{z}=&A({{\bar{x}}_{n}},\mu )z +B_1 \mu f_p {{k}_{\rho}}({e}/{\rho })(\rho -{{\rho }_{\infty }}) \\
&+ B_1 \mu [1+(f_p-1)f_t]{{g}_{1}}({{x}_{1}})\left( {{\eta }_{2}}+{{x}_{2,e}} \right)) +{{D}_{\mu }}\omega 
\end{aligned}
\end{equation}
where the definitions of  $A({{\bar{x}}_{n}},\mu )$ and ${{D}_{\mu }}=D_{\mu }^{1}$  are the same as those in Proposition \ref{Proposition_Performance}. 
 
Similar to the fashion in Section III.A, three compact sets and their relationship are illustrated in Fig. \ref{fig:Fig_explain2}.

\emph{User-defined invariant set}: define the region $\Omega =\left\{ z|\Phi={{z}^{T}}D_{\mu }^{{\nu }/{2}\;}PD_{\mu }^{{\nu }/{2}\;}z \le 1 \right\}$, $P>0$, as an invariant set.

\emph{User-defined safety region}: define the region ${{\Omega }_{0}}=\left\{ z|\Phi \le {{\Phi}_{0}} \right\}$ as a safety region. The user-defined variable $0<{{\Phi}_{0}}<1$ determines the safety region $\Omega_0$'s size. 

\emph{User-defined dead zone}: define the region ${{\Omega }_{\varepsilon}}=\left\{ z| \Phi > {{\Phi}_{0}} \right.$   $ \left. \cap \vert e \vert  /\rho\le \varepsilon \right\}$ as a dead zone, where $\varepsilon>0$ determines the width of ${{\Omega }_{\varepsilon}}$.  The reason for the existence of this dead zone is due to the term $\dot{\rho }({e}/{\rho })$ in the derivative of $z_1$ goes to 0 along $e$ despite affected by $\dot{\rho}$. On the contrary, if this dead zone is not divided, according to the design in $(\ref{Eq_4_4})$, the derivative of $\rho$ would become very large and its sign would also change with $e$.

\begin{figure}[t]
\centerline{\includegraphics[width=0.22\textwidth]{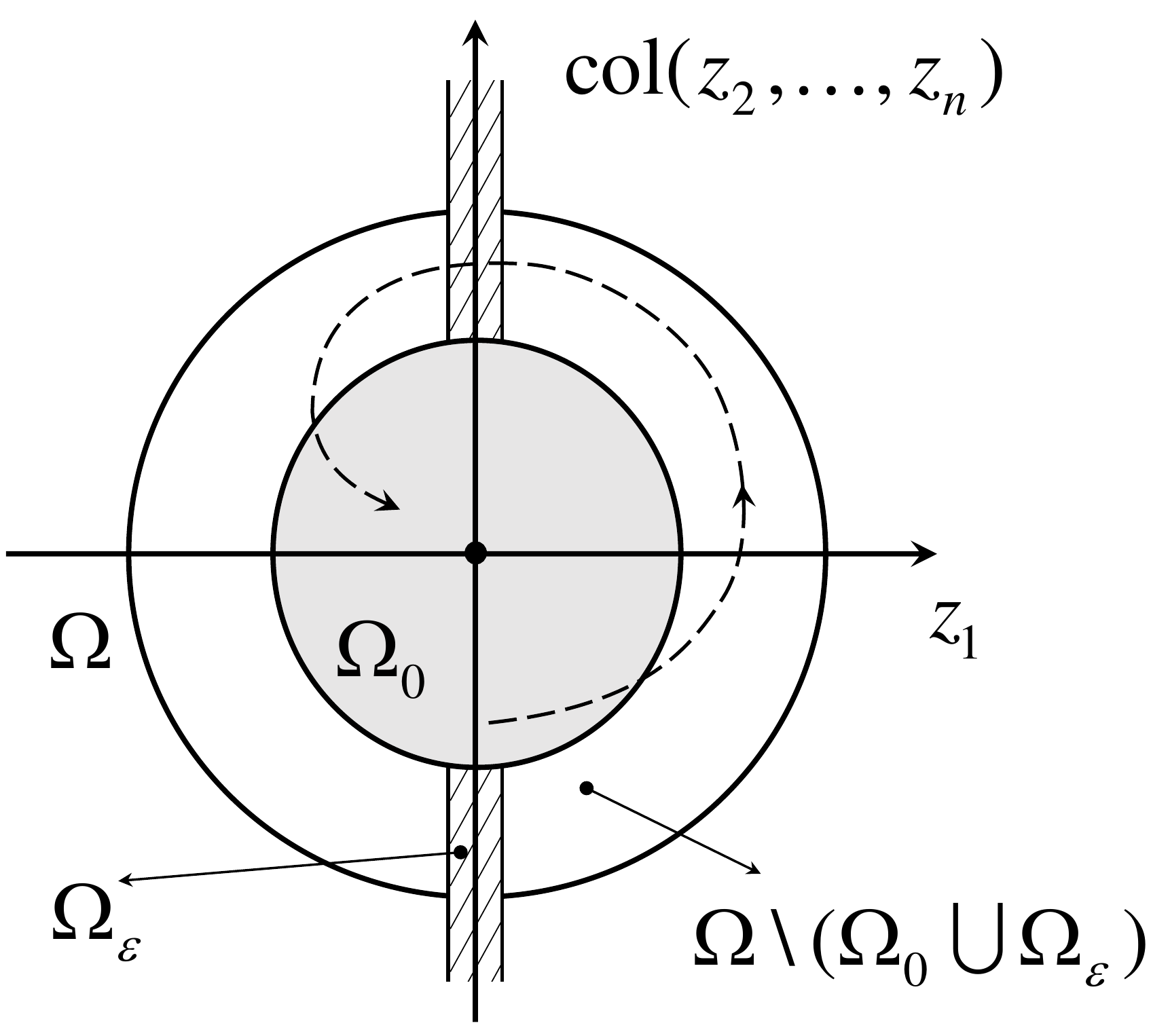}}
\caption{Illustration of invariant set $\Omega$, safety region $\Omega_0$ and dead zone $\Omega_\varepsilon$.}
\label{fig:Fig_explain2}
\end{figure} 

Henceforth, we provide the PSE $f_p(z,\mu)$ for the balance between the performance constraint and effects from DEVs and the DZT $f_t(e,\rho)$ for the dead zone transition.

\begin{align}  
& f_p(z,\mu)= \left\{ \begin{matrix}
\begin{aligned}
&  0,  \ \ \ \ \ \ \ \ \ \ \ \ \ \ \ \  \ \ \ \ \ \  \text{if}\  \Phi \ge 1  \\
& {(1-\Phi)}/{(1-{{\Phi}_{0}})},  \ \ \text{if}\ \Phi_0<\Phi<1  \\
&  1,   \ \ \ \ \ \ \ \ \ \ \ \ \ \ \ \  \ \ \ \ \ \  \text{if}\ \Phi \le {{\Phi}_{0}}  \\
\end{aligned}
\end{matrix} \right.   \label{Eq_4_8a}  \\
& f_t(e,\rho)=\left\{ \begin{matrix}
\begin{aligned}
&0,  \ \ \ \ \ \ \ \ \ \ \ \ \ \ \  \ \ \  \  \text{if}\ \vert {e}\vert/\rho \le  0.5\varepsilon \\
&(\vert {e}\vert/\rho-\varepsilon )/\varepsilon  ,\ \ \ \ \   \text{if}\ 0.5\varepsilon <\vert {e}\vert/\rho<  \varepsilon \\
&1, \ \ \ \ \ \ \ \ \ \ \ \ \ \ \ \  \ \ \  \text{if}\ \vert {e}\vert/\rho>  \varepsilon
\end{aligned}
\end{matrix} \right. \label{Eq_4_8b}
\end{align}

\subsection{Stability analysis}

\begin{theorem} \label{Theorem_BPC}
Consider the nonlinear system in form of (\ref{Eq_2_1}) under Assumptions \ref{Assumption_1}-\ref{Assumption_3}. With the auxiliary system (\ref{Eq_4_4}), the control law is designed according to (\ref{Eq_4_6}). The dynamic of performance function is given in the first equation of (\ref{Eq_4_4}). The
PSE $f_p(z,\mu)$ and DZT $f_t(e,\rho)$ are specified as (\ref{Eq_4_8a}) and (\ref{Eq_4_8b}), respectively.  Then, for the initial condition $s(0)=z(0) \in \Omega$, we have the following conditions: (i) under the premise $\vert e \vert/\rho \ge \varepsilon$, the compensated tracking error $z$ is always contained in the invariant set $\Omega$; 
(ii) the performance constraint $-\underline{\delta }\rho (t)<e(t)<\bar{\delta }\rho (t)
$ is guaranteed, meanwhile, $\rho$ finally converges to $\rho_\infty$;
(iii) the solutions of the closed-loop system are uniformly bounded.
\end{theorem}
 
\begin{proof}
The proof is divided into three parts according to the inclusion relationship: (i) $z \notin \Omega $ and $\vert e \vert/\rho \ge \varepsilon$; (ii) $z \in \Omega_0 $; (iii) $z \in \Omega_{\varepsilon}$. Their corresponding presentations are provided in the following cases. 

\emph{Case 1.} $z \notin \Omega $ and $\vert e \vert/\rho \ge \varepsilon$ which bring $f_p=0$ and $f_t=1$.  The system trajectory runs out of the region $\Omega $ and converts the expression (\ref{Eq_4_8}) into the simplified form
$\dot{z}=A({{\bar{x}}_{n}},\mu )z+{{D}_{\mu }}\omega$. 
Following the development of Proposition \ref{Proposition_Performance}, the invariant set $\Omega$ for the performance-constrained control is thus established. Consequently, since $z$ is always inside  $\Omega$, the performance constraint is obvious guaranteed.

\emph{Case 2.} $z \in \Omega_0 $ which brings $f_p=1$,  $\dot{\rho}=-{{k}_{\rho}}(\rho -{{\rho }_{\infty }})$. The evaluation of system performance is in a safe state, that is, the DEVs are still in a tolerable range and $\rho$ finally converges to $\rho_\infty$. The detailed discussion is in a similar fashion to the previous Case 2 in the proof of Theorem \ref{Theorem_BCF}, which is omitted for concision.

\emph{Case 3.} $z \in \Omega_{\varepsilon}$ which brings $f_t=0$, $\dot{\rho}=-f_p{{k}_{\rho}}(\rho -{{\rho }_{\infty }})$ due to (\ref{Eq_4_4}).
Consider now the special dynamics form of $z$
\begin{equation} \label{Eq_4_11}
\begin{aligned}
\dot{z}=&A({{\bar{x}}_{n}},\mu )z +\mu {{B}_{1}}(f_p{{k}_{\rho}}({e}/{\rho })(\rho -{{\rho }_{\infty }}) \\
&+{{g}_{1}}({{x}_{1}})({{\eta }_{2}}+{{x}_{2,e}}) )+{{D}_{\mu }}\omega.
\end{aligned}
\end{equation}
Obviously, we can reach the following consensus: (i) $z_1$ belongs to a small vicinity of the origin, which means that $e$ is within the constraint range of $\rho$;
(ii) it is impossible to violate the performance constraints $-\underline{\delta }\rho (t)<e(t)<\bar{\delta }\rho (t)$ ($z_1 \to \pm \infty$) without leaving the dead zone $\Omega_\varepsilon$. In this sequel, we get the inference that the system trajectory will not be trapped in the dead zone $\Omega_\varepsilon$, and will inevitably return to the region $\Omega \setminus (\Omega_0 \cup \Omega_\varepsilon)$ or the invariant set $\Omega_0$.

The above facts show the balance operation between the performance and saturation constraints.
\end{proof} 

\begin{figure}[t]
\centering
\begin{minipage}{0.49\linewidth} 
  \centerline{\includegraphics[width=4.65cm]{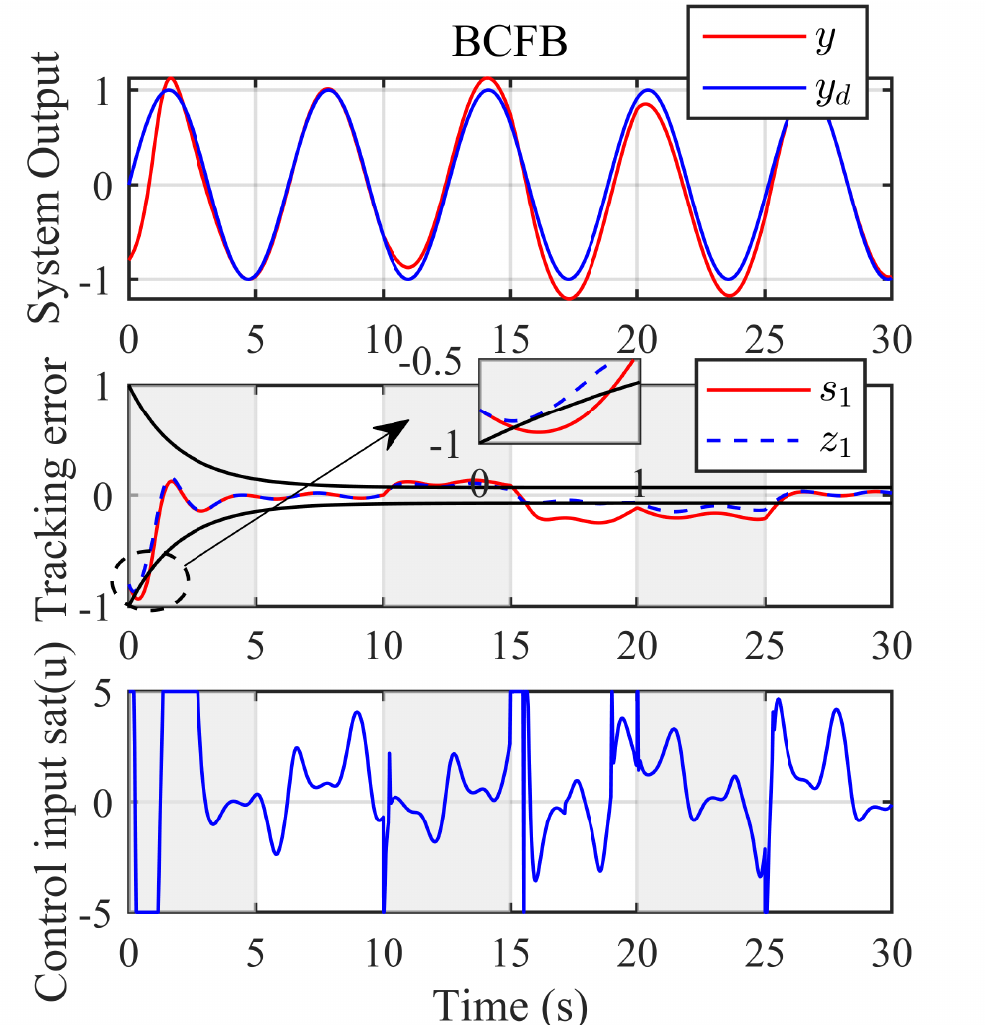}}
\end{minipage}
\hfill
\begin{minipage}{0.49\linewidth} 
  \centerline{\includegraphics[width=4.65cm]{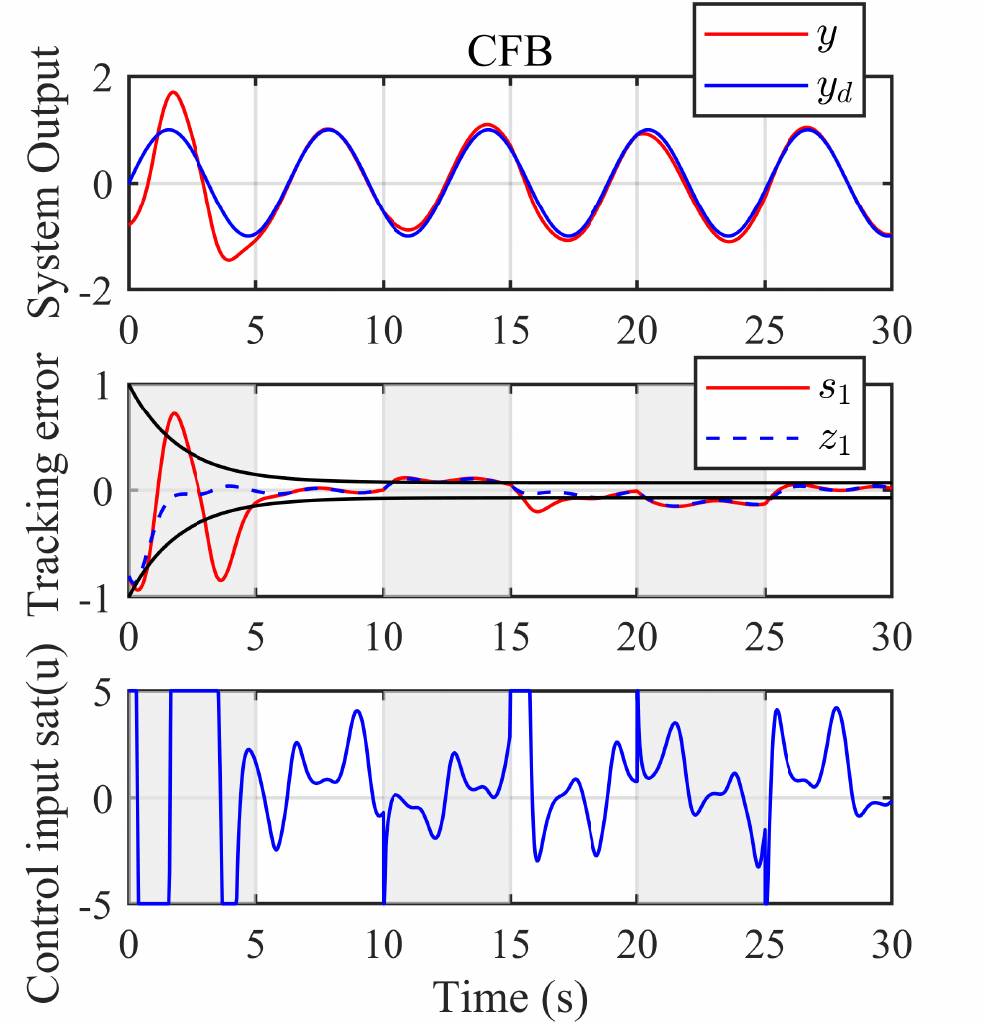}}
\end{minipage}
\caption{Case A: Tracking performance comparisons between BCFB and CFB including the system output, tracking error and control input.}
\label{CaseA_1}
\end{figure}

\begin{figure}[tbp]
\centering
\begin{minipage}{0.49\linewidth} 
  \centerline{\includegraphics[width=4.65cm]{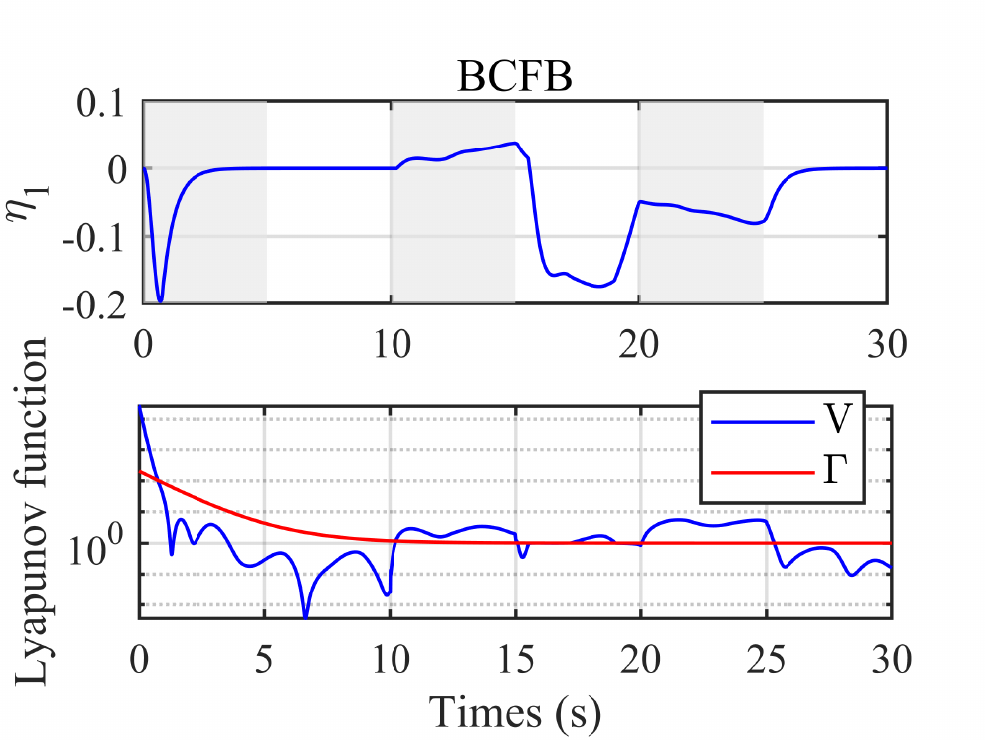}}
\end{minipage}
\hfill
\begin{minipage}{0.49\linewidth} 
  \centerline{\includegraphics[width=4.65cm]{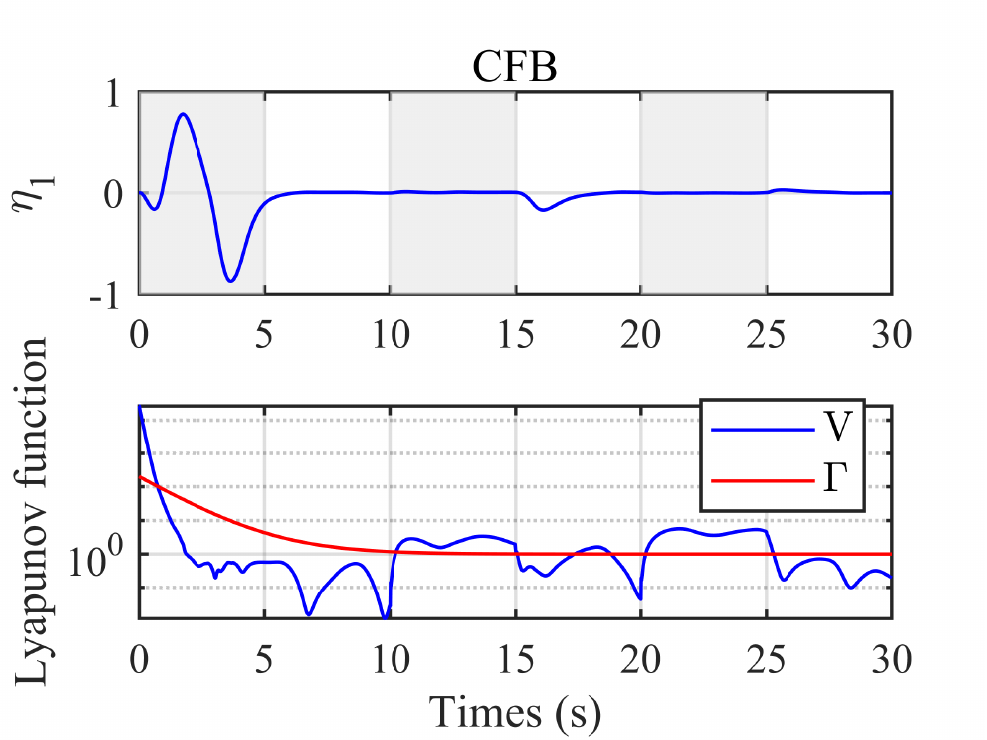}}
\end{minipage}
\caption{Case A: Balance operation comparisons between BCFB and CFB including the compensation signal $\eta_1$ and the Lyapunov function $V$.}
\label{CaseA_2}
\end{figure}

\begin{figure}[tbp]
\centering
\begin{minipage}{0.49\linewidth} 
  \centerline{\includegraphics[width=4.65cm]{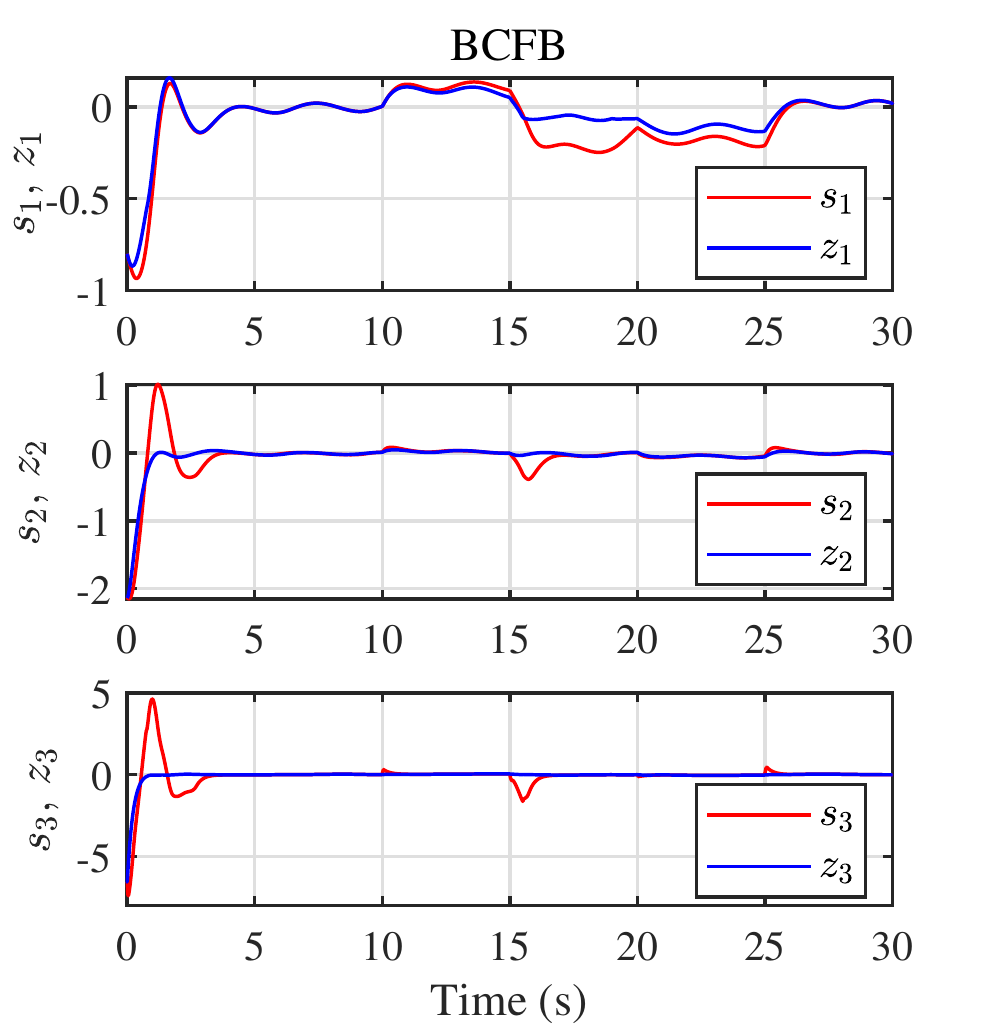}}
\end{minipage}
\hfill
\begin{minipage}{0.49\linewidth} 
  \centerline{\includegraphics[width=4.65cm]{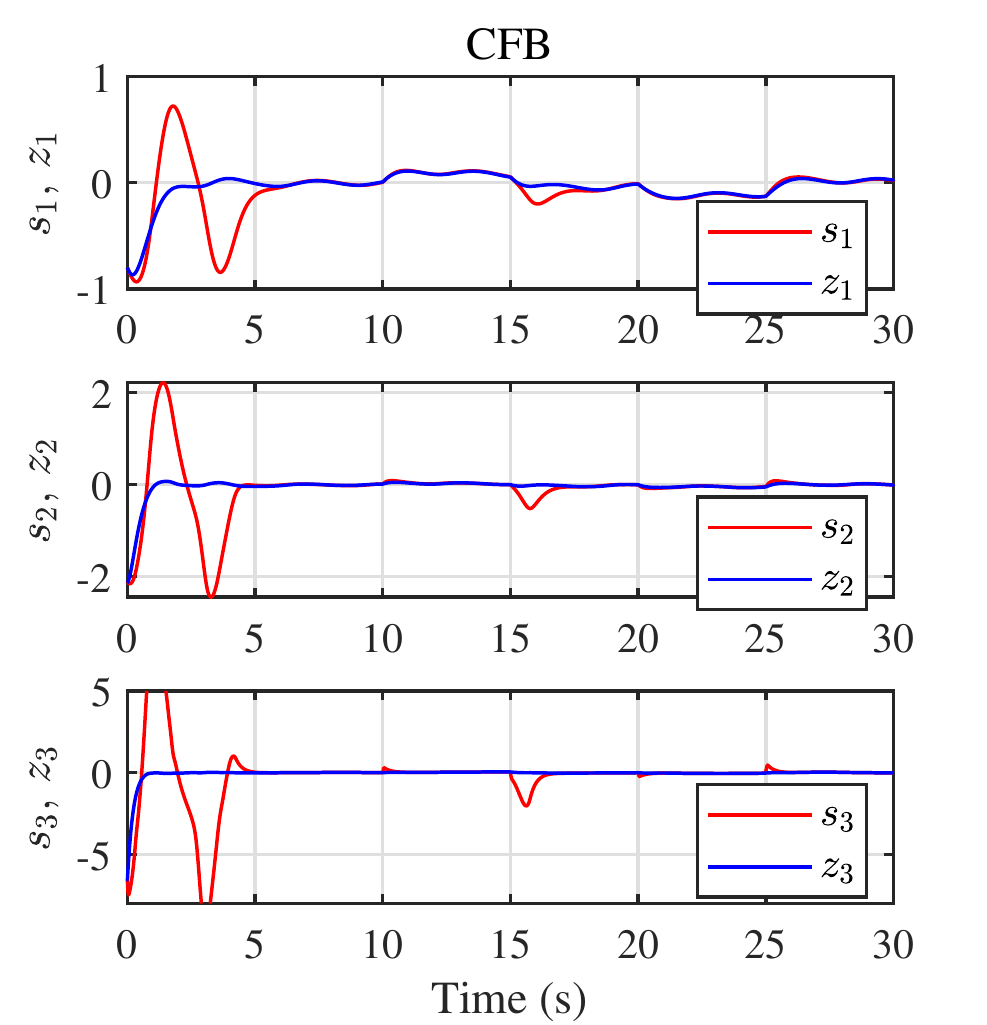}}
\end{minipage}
\caption{Case A: Command filter error comparison between BCFB and CFB.}
\label{CaseA_3}
\end{figure}

\begin{remark}
In FPC \cite{yong2020flexible}, the saturations are fed back to the performance functions through a novel auxiliary system, consequently, the resulting performance function can degrade the user-specified performance constraints when saturation happens and can recover it when  no saturation occurs. We summarize the following differences between FPC and BPC:
\begin{enumerate} [(i)]
\item  For handling the input saturation, the FPC splits the impact of saturation and passes it on to the upper and lower performance functions.
However, the proposed BPC adopts the adaptive form of performance bounds without compensation signals;

\item  The FPC ignores the system's own adjustment ability, which means that the system is of some tolerance to saturation. In BPC, a PSE parameter is defined for the evaluation of the system safety 
to decide whether to do compensation for the operation of 'balance'; 
\item  The FPC has not been able to realize the compensation of the filter errors from the dynamic surface control. 
When encountering sudden disturbance, the chatter would happen and threaten performance constraints.
Consequently, the proposed BPC feeds back the filter errors into the performance function design.
\end{enumerate} 
\end{remark}

\section{Numerical examples}

 \begin{figure}[tp]
\centering
\begin{minipage}{0.49\linewidth} 
  \centerline{\includegraphics[width=4.65cm]{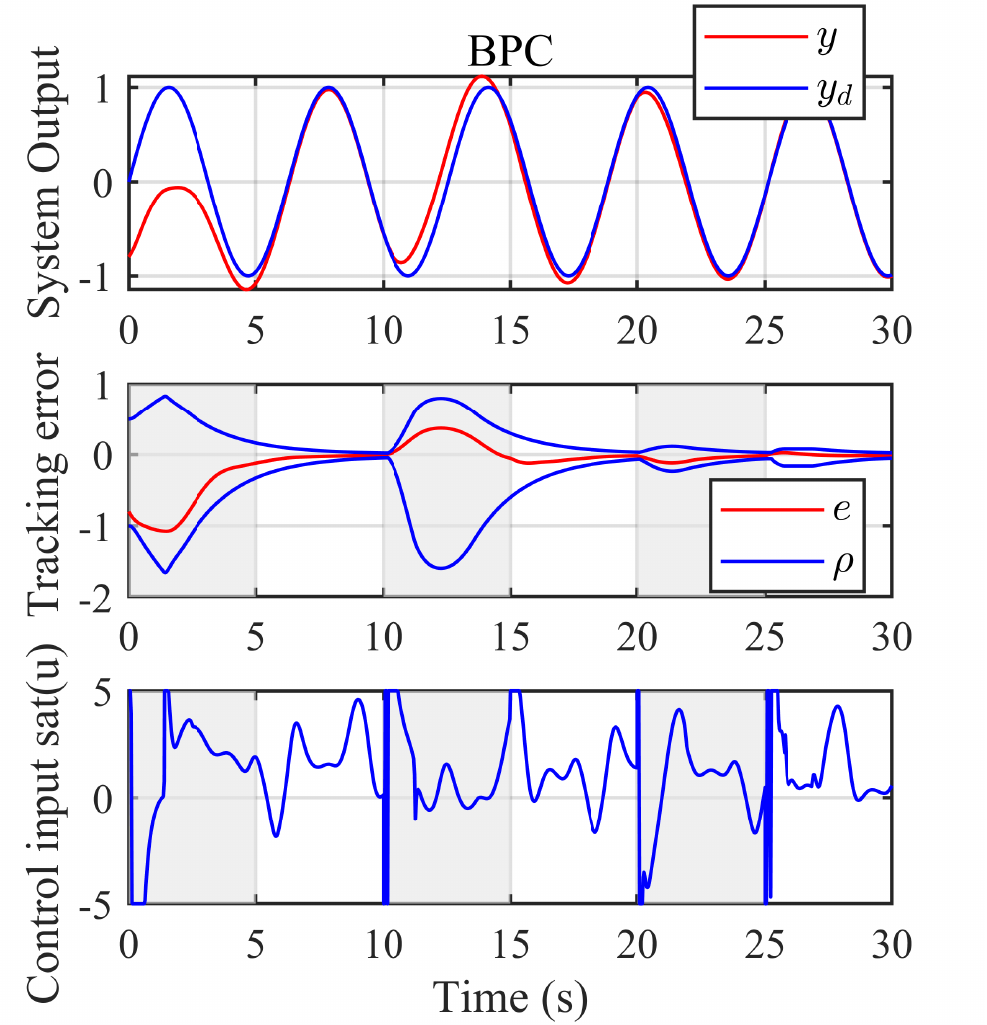}}
\end{minipage}
\hfill
\begin{minipage}{0.49\linewidth} 
  \centerline{\includegraphics[width=4.65cm]{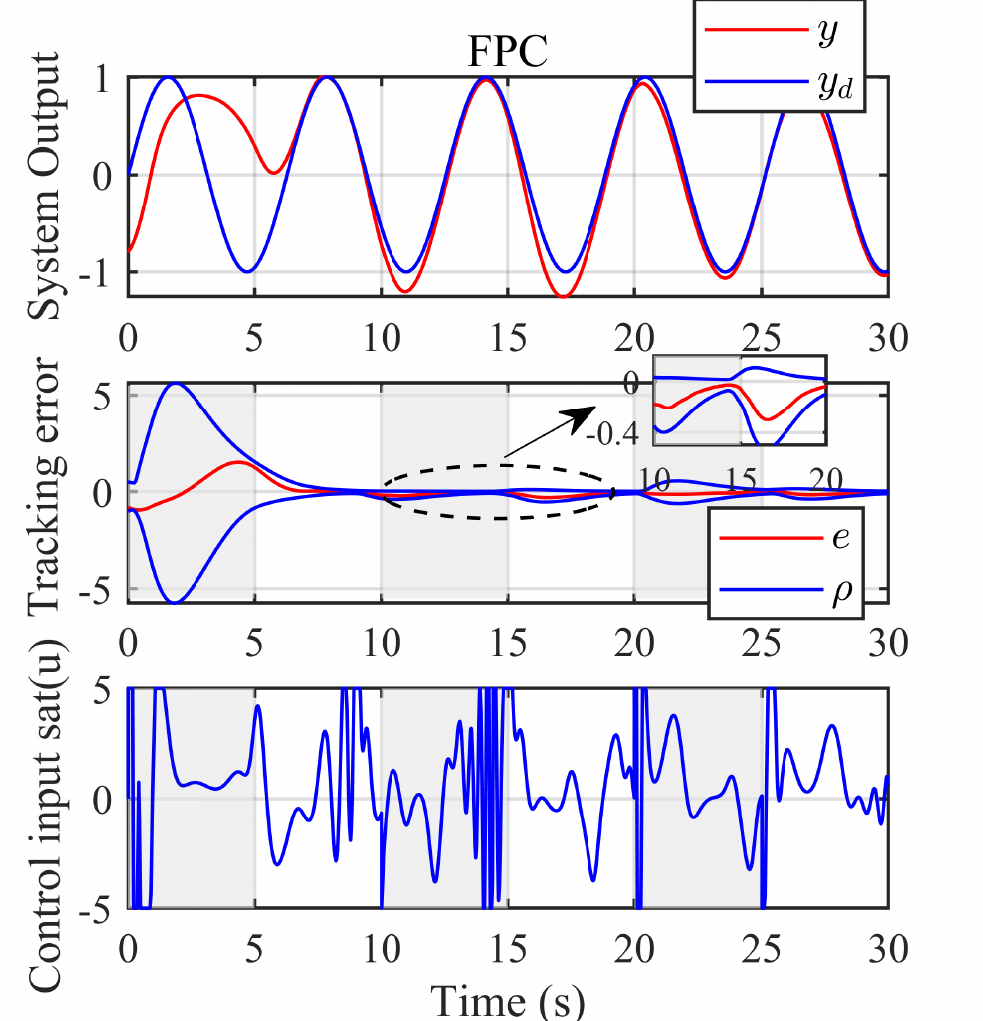}}
\end{minipage}
\caption{Case B: Tracking performance comparisons between BPC and FPC including the system output, tracking error and control input.}
\label{CaseB_1}
\end{figure}

\begin{figure}[tbp]
\centering
\begin{minipage}{0.49\linewidth} 
  \centerline{\includegraphics[width=4.65cm]{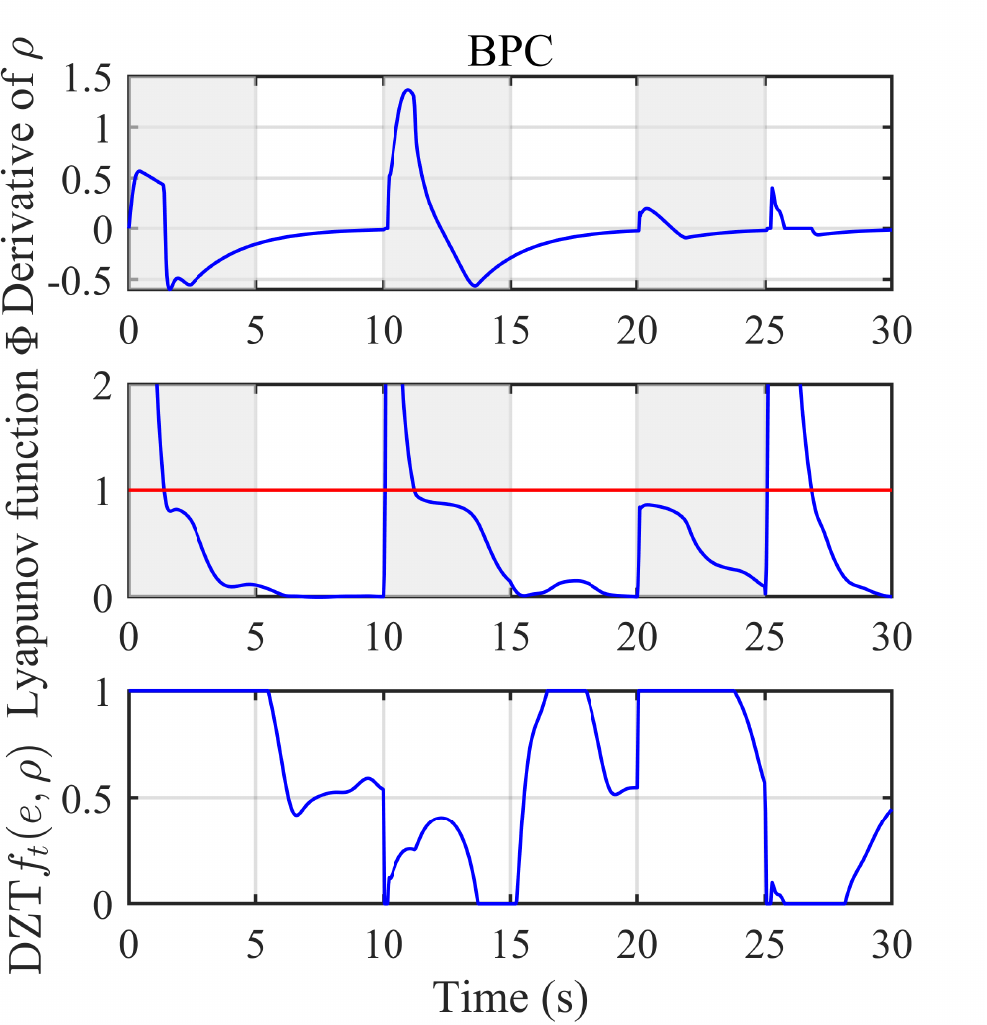}}
\end{minipage}
\hfill
\begin{minipage}{0.49\linewidth} 
  \centerline{\includegraphics[width=4.65cm]{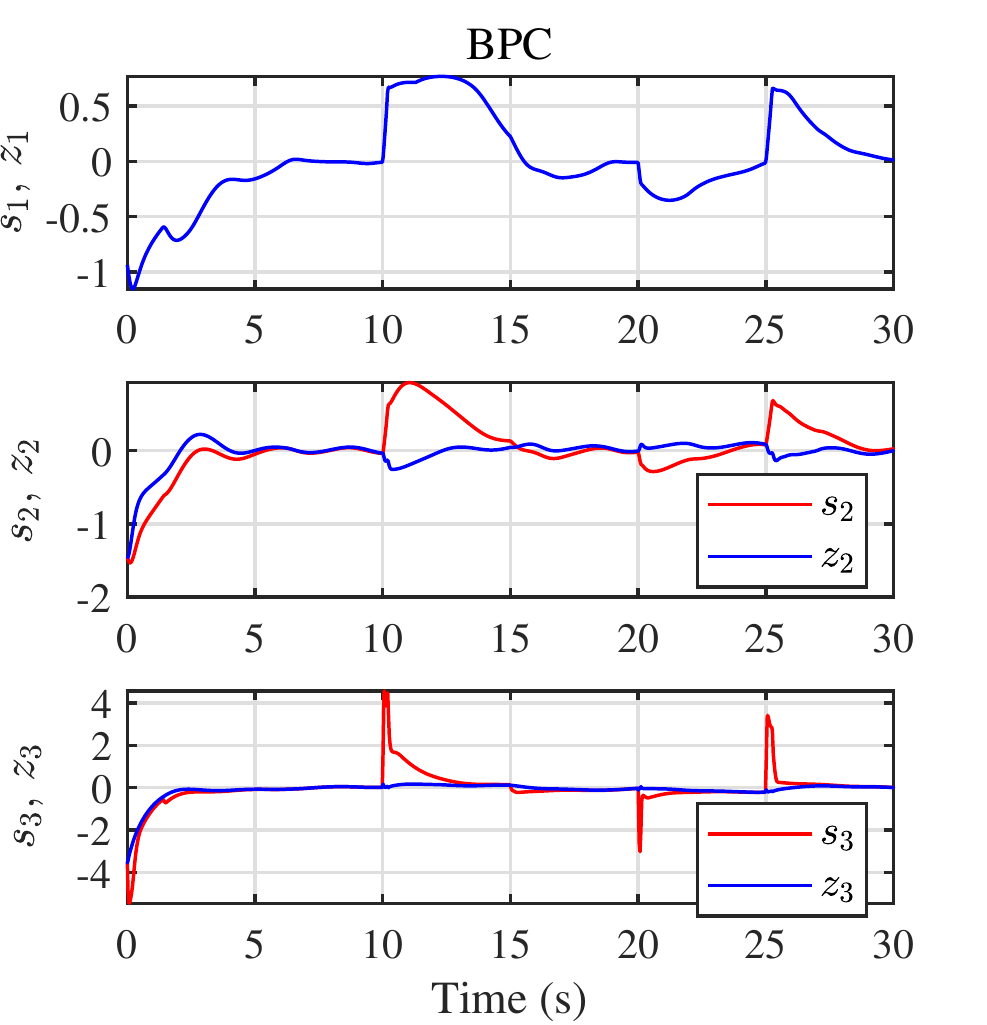}}
\end{minipage}
\caption{Case B: (Left) Balance operations of BPC including the derivative of performance function $\rho$, the Lyapunov function $\Phi$ and the parameter DZT; (Right) Command filter error comparisons  between $s_i$ and $z_i$ in BPC.}
\label{CaseB_2}
\end{figure}

Consider the system (\ref{Eq_2_1}) for $n=3$.
$f_1(x_1)={\sin ({{x}_{1}})}/{(1+x_{1}^{2})}\;$, $f_2(\bar{x}_2)=\tanh ({{x}_{2}})\exp (-{{({{x}_{1}}{{x}_{2}})}^{4}})$ and $f_3(\bar{x}_3)={{x}_{1}}{{x}_{2}}$ are known smooth functions; $g_1(x_1)=g_2(\bar{x}_2)=g_3(\bar{x}_3)=1$ are constant; $\omega_1(t)=0.1\sin ({{x}_{1}})\cos (t)$, $\omega_2(t)=0.15\sin ({{x}_{1}}{{x}_{2}})$ and $\omega_3(t)=0.1\cos ({{x}_{3}})\sin (t)$ are time-varying disturbances. The control constraint is described by $u_{\text{max}}=-u_{\text{min}}=5$. 
The initial condition is $x(0)={{\left[ -0.8,0.9,0.1 \right]}^{T}}$. The desired trajectory is $y_r=\sin(t)$.

\emph{Case A.}(BCFB vs CFB) To illustrate the effectiveness and merits of the proposed BCFB, the simulation and comparison between BCFB (\ref{Eq_3_2_2})-(\ref{Eq_3_2_4}) and CFB discussed in Remark \ref{Remark_CFB} 
are presented in Fig. \ref{CaseA_1}-\ref{CaseA_2}. The designed parameters of the controller are chosen as $k_1=2$, $k_2=3$, $k_3=4$, $\tau_1=\tau_2=0.01$. The performance constraint for BCFB is characterized by $\kappa=0.5$, $\rho_0=1$, $\rho_\infty=0.1$ and the safety region of BCFB is determined by $\sigma=0.9$. Additional disturbances $\omega_a(t)=0.2\text{exp}({-0.1(t-t_j)})$ interferes with the system when $t_j=10,15,20,25\,\text{s}$, i.e., $\omega_i(t):=\omega_i+\omega_a$ for $i=\mathbb{I}_{1:3}$. It is obvious that the tracking BCFB is improved with smaller tracking errors and fewer response fluctuations, and the problem of input saturation gets alleviated. The coincidence of the tracking error curves $s_1$ and $z_1$ in BCFB is obvious larger than that in CFB. The curves of $\eta_1$ in CFB and BCFB are quite different, which is the main source contributing to the different performances. Accordingly, the Lyapunov function curve in CFB does not truly reflect the  tracking performance of the system. Given the performance bounds, the BCFB balances the contradiction between the performance and saturation and converges the tracking error to the predetermined performance range without PPC control. 

\emph{Case B.}(BPC vs FPC) When it comes to the performance constrained control (BPC, FPC), the designed parameters of the controller are chosen as $k_1=k_2=k_3=1$, $\tau_1=\tau_2=0.01$, $\nu=0$, the initial value and convergence rate of $\rho$ is $\rho_0=1$, $k_\rho=0.5$, respectively. The safety region of BPC is determined by $\Phi_0=0.7$.  Additional disturbances $\omega_a(t)=0.1\text{exp}({-0.1(t-t_j)})$ interferes with the system when $t_j=10,15,20,25\,\text{s}$, i.e., $\omega_i(t):=\omega_i+\omega_a$ for $i=\mathbb{I}_{1:3}$. 
Note that, in Fig. \ref{CaseB_1}, both BPC and FPC achieve flexible performance bounds which avoid saturation damage to system stability. Furthermore, it can be observed that (i) due to the lack of filtering error processing in the FPC, the obtained control input is more susceptible to the external noise; (ii) at the initial moment due to the occurrence of saturation, the performance function in the FPC has greater fluctuations to compensate for the saturation, while the performance function in the BPC has already adjusted the system evaluation to safety at 1.425s; (iii) the fluctuation of the performance function in BPC at 10-15s comes from the tracking error entering the dead zone $ \Omega_{\varepsilon}$.
In Fig. \ref{CaseB_2}, by the action of the proposed BPC, when the system state is outside the invariant set $\Omega$ ($\Phi(t)>1$) due to the sudden disturbance, $\rho$ changes to adapt to the influence of DEVs. Until the system state returns to the invariant set $\Omega$ or the safe zone $\Omega_0$, $\rho$ will be re-driven to converge to $\rho_\infty$.

 Overall, the effectiveness of the proposed balanced control between the performance and saturation has been illustrated.

\section{Conclusion}
In our work, the common constrained control problem for nonlinear systems has been addressed and a balanced control methodology has been presented for the balance between performance and saturation through introducing the concept of the PSE. The invertible DEVs, which threaten the performance of CFB and the safety of PPC, get fully utilized and compensated, yielding the resultiant BCFB and BPC. Through the balance operation according to PSE, (i) in BCFB, the performance-related invariant set drives the tracking error $s_i$ and $z_i$ more consistent; (ii) in BPC, the performance constraint is flexible with the system safety evaluation. 

\appendices
\section{Proof of Proposition 1}
The detailed proof contains the following two parts:

\emph{Performance constraint:} $V(t) \le \Gamma(t) \Rightarrow\vert Cz\vert \le \rho(t)$. 

Given the change of coordinates $\bar{z}=P^{1/2}z$, the Lyapunov function becomes $V(t)=\bar{z}^T\bar{z}$. Likewise, the boundary of performance constraint $\vert  Cz\vert = \rho $ becomes the hyperplane $\vert CP^{-1/2} \bar{z}\vert = \rho$. As a result, the maximum admissible value of the Lyapunov function is the square distance between this hyperplane and the origin, $\Gamma(t)=\Vert \bar{z} \Vert ^2=(\rho(/ \Vert CP^{-1/2}\Vert)^2 $, which is equivalent to  $\Gamma(t)=\rho^2/(CXC^T)$.

\emph{Invariance:} 
Introduce another Lyapunov function candidate $M(t)=V(t)/\Gamma(t)$ whose time derivative yields
\begin{equation} \label{Eq_prop1_1}
\frac{\text{d} M}{\text{d} t}
\le  \frac{1}{\Gamma}(z^T(PA+A^TP+2\kappa P)z+2z^TP\omega).
\end{equation}
Invoking the condition $g_{i,M} > \vert g_i({{{\bar{x}}}_{i}}) \vert > g_{i,m}$ given in Assumption \ref{Assumption_1}, there exist a positive constant  $\epsilon$ and a feasible positive-definite matrix $P$ making $P{{A}_{g}}+A_{g}^{T}P \le \epsilon P$ hold, from which the relation (\ref{Eq_prop1_1}) becomes
\begin{equation}   \label{Eq_prop1_2}
\frac{\text{d} M}{\text{d} t}\le \frac{1}{\Gamma}{{\left[ \begin{matrix}
   z  \\
   \omega   \\
\end{matrix} \right]}^{T}}H\left[ \begin{matrix}
   z  \\
   \omega   \\
\end{matrix} \right]+\frac{\alpha}{\Gamma}\left( {{\omega }^{T}}W\omega -{{z}^{T}}Pz \right)
\end{equation}
where the parameter $\alpha>0$, the matrix 
\[H=\left[ \begin{matrix}
   P{{A}_{0}}+A_{0}^{T}P+(\epsilon +\alpha +2\kappa)P & P  \\
   P & -\alpha W  \\
\end{matrix} \right]\]
is negative semidefinite due to (\ref{Eq_2_4}) and the Schur complement. Since ${{z}^{T}}Pz=\Gamma$ and ${{\omega }^{T}}W\omega \le \Gamma_\infty \le\Gamma$, the invariant set $\Omega$ is thus established.

\section{Proof of Proposition 2}
The derivative of Lyapunov function candidate $V$ along the solution of (\ref{Eq_2_5}) gives
\begin{equation}
\begin{aligned}
  & \frac{\partial V}{\partial z}(Az+\omega )\le {{\left[ \begin{matrix}
   D_{\mu }^{{\nu }/{2}\;}z  \\
   D_{\mu }^{1-{\nu }/{2}\;}\omega   \\
\end{matrix} \right]}^{T}}H\left[ \begin{matrix}
   D_{\mu }^{{\nu }/{2}\;}z  \\
   D_{\mu }^{1-{\nu }/{2}\;}\omega   \\
\end{matrix} \right] \\ 
 & \ \   +\alpha \left( {{\omega }^{T}}D_{\mu }^{1-{\nu }/{2}\;}WD_{\mu }^{1-{\nu }/{2}\;}\omega -{{z}^{T}}D_{\mu }^{{\nu }/{2}\;}PD_{\mu }^{{\nu }/{2}\;}z \right) \\ 
\end{aligned}
\end{equation} 
is negative semidefinite due to (\ref{Eq_2_4}) and the Schur complement. Since ${{z}^{T}}D_{\mu }^{{\nu }/{2}\;}PD_{\mu }^{{\nu }/{2}\;}z=1$ and ${{\omega }^{T}}D_{\mu }^{1-{\nu }/{2}\;}WD_{\mu }^{1-{\nu }/{2}\;}\omega \le 1$, the invariant set is thus established.

\small
\bibliographystyle{ieeetr}
\bibliography{ref}
 
\end{document}